\documentclass[11pt,reqno]{article}

\usepackage{amssymb,amscd,amsmath,amsthm}
\numberwithin{equation}{section}
\usepackage{graphicx}
\usepackage{tikz}
\usetikzlibrary{arrows.meta,calc,decorations.pathmorphing,positioning}
\newtheorem{theorem}{Theorem}[section]
\newtheorem{proposition}[theorem]{Proposition}

\newtheorem{corollary}[theorem]{Corollary}
\newtheorem{definition}[theorem]{Definition}
\newtheorem{remark}[theorem]{Remark}

\def\cb{{\mathcal B}}

\def\ce{{\mathcal E}}

\def\cw{{\mathcal W}}

\def\bh{{\mathbb H}}

\def\tr{{\rm Tr}}
\def\L{\Lambda}
\def\G{\Gamma}
\def\ce{\mathcal E}

\def\ffi{\varphi}

\def\Tr{\mathrm{Tr}}
\def\<{\langle}
\def\>{\rangle}

\def\1{\mathbf{1}}

\def\cw{\cal W}

\def\cal{\mathcal}

\def\bh{\mathbf{h}}

\def\id{{\bf 1}\!\!{\rm I}}

\begin{document}

\begin{center}
{\Large {\bf Quantum Markov Chains for an Asymmetric Mixed Ising–XY Model on a Cayley Tree}}\\[1cm]
{\large   {\sc Farrukh Mukhamedov$^{1,a}$}}\\[2mm]
\end{center}
$^1$ Department of Mathematical Sciences, College of Science, \\
United Arab Emirates University,\\
 P.O. Box 15551, Al Ain,  Abu Dhabi, UAE\\[0.4cm]
{\small
$^a$ e-mail: {\tt  far75m@yandex.ru; farrukh.m@uaeu.ac.ae}}

\begin{abstract}
We study a mixed quantum Ising--$XY$ model on the semi-infinite rooted Cayley tree of order two. For every vertex $u$, the edge $\langle u,(u,1)\rangle$ carries an $XY$ interaction and the edge $\langle u,(u,2)\rangle$ carries an Ising interaction. Using the compatibility criterion for tree-indexed quantum Markov chains and consistently working with the normalized trace, we derive the translation-invariant boundary equation and compute explicitly the associated local transfer operator, namely the one-step partial-trace map which propagates successor boundary data to the parent vertex. We prove that the boundary equation has a unique positive translation-invariant solution for all $J_I,J_{XY}\in\mathbb R$ and $\beta>0$. Hence the model admits a unique translation-invariant quantum Markov chain generated by a positive translation-invariant boundary condition. We also show that the reduced boundary-law dynamics, i.e. the induced finite-dimensional recursion for the boundary-law parameters, has no admissible periodic points of period greater than one and compute the local two-site entanglement on the natural three-site cluster of the tree.\\
\textit{Key Words}:  Mixed Ising–XY Model; Cayley Tree;  entanglement; quantum Markov chain.
\end{abstract}

\section{Introduction }\label{intr}

Quantum Markov chains and quantum Markov states provide a natural non-commutative extension of classical Markov fields and Gibbs measures. They have been studied on one-dimensional spin chains, on $\mathbb Z^d$-lattices, and on graph-indexed systems; see, for example, \cite{[Ac74f],[AcFi01a],AF03,[AcFiMu07],AOM,fannes2,G,ILW,Kum,SA26}. Among non-amenable graphs, Cayley trees occupy a special place. Their boundary remains comparable with the volume in the thermodynamic limit, while their recursive structure often allows exact calculations that are inaccessible on Euclidean lattices \cite{AMSa,AMSa2,S23,SM24}. This makes tree-indexed quantum Markov chains a convenient setting for studying how boundary laws and local transfer operators determine infinite-volume states.

Throughout the paper, a \emph{boundary law} means the family of positive single-site operators attached to the vertices and used as boundary weights in the finite-volume states. A \emph{local transfer operator} is the one-vertex map obtained by conjugating with the interaction operator at a parent and then taking the normalized partial trace over the two successors; the boundary-law equation is precisely the fixed-point, or recursive, equation for these local maps. After a boundary operator is written in Pauli (Bloch) coordinates, this operator recursion becomes a finite-dimensional real recursion. We call this induced recursion the \emph{reduced boundary-law dynamics}.

For nearest-neighbor spin models on a tree, the contrast between Ising-type and $XY$-type interactions is already visible at the level of translation-invariant boundary conditions. Classical and quantum Ising models on Cayley trees may exhibit phase coexistence and multiple boundary laws \cite{AMSa3,MBS161,MR04,MS19,MS21,Roz}, whereas for the quantum $XY$ model on the Cayley tree of order two the positive translation-invariant boundary law is unique \cite{MG17,MG19}. This raises a natural question: what happens if one combines diagonal Ising interaction and coherent exchange-type $XY$ interaction within the same tree model?

In the present paper we study precisely this mixed situation. We consider the semi-infinite rooted Cayley tree of order two and assign, for every vertex $u$, an $XY$ interaction to the edge $\langle u,(u,1)\rangle$ and an Ising interaction to the edge $\langle u,(u,2)\rangle$. Related mixed Ising--$XY$ systems on trees and other hierarchical lattices have been investigated in \cite{AAM25,ASH14,BO,MBSG20,ZM17}. Our aim is to construct the corresponding tree-indexed QMC by the compatibility criterion of \cite{AMSOa3}, derive the translation-invariant boundary equation, and analyze all positive translation-invariant solutions.

The main point of the paper is that, the translation-invariant boundary equation admits exactly one positive solution for all $J_I\in\mathbb R$, $J_{XY}\in\mathbb R$, and $\beta>0$. Consequently, the mixed model considered here has a unique translation-invariant QMC associated with a positive translation-invariant boundary condition. We also prove that the reduced dynamical system has no admissible periodic points of period greater than one, compute the local three-site density matrix of the translation-invariant state, and derive explicit pairwise entanglement measures on the natural three-site cluster.

The model studied in \cite{AAM25} is a level-alternating mixed Ising--$XY$ model: the Ising interaction acts on odd levels of the Cayley tree, while the \(XY\) interaction acts on even levels. In contrast, the present paper considers a locally asymmetric oriented model in which, for every vertex \(u\), the edge \(\langle u,(u,1)\rangle\) carries the \(XY\) interaction and the edge \(\langle u,(u,2)\rangle\) carries the Ising interaction. Hence the asymmetry is built into the branching geometry itself rather than into level parity.

This geometric difference leads to a different boundary-law equation and a different recursive structure. While \cite{AAM25} proves uniqueness of the periodic translation-invariant QMC and also constructs non-translation-invariant QMCs, the present paper proves a stronger rigidity result for the oriented model: the positive translation-invariant boundary condition is unique for all values of the parameters, the reduced dynamics has no admissible periodic points of period greater than one, and every positive solution of the full two-child boundary equation is scalar. Consequently, the non-translation-invariant positive scalar boundary laws that remain do not generate new states, but only the same QMC as the unique translation-invariant positive solution.

The paper is organized as follows. Section 2 recalls the basic notation for the Cayley tree and the definition of a backward quantum Markov chain. In Section 3 we construct the finite-volume operators and state the compatibility criterion for the mixed model. Section 4 derives the corrected translation-invariant boundary equation and proves uniqueness of the positive solution. Section 5 discusses the induced dynamical interpretation on normalized Bloch variables. Section 6 computes local two-site entanglement on the natural three-site cluster. Section 7 contains a brief discussion of possible applications. We notice that the considered model opens new insight into the hidden quantum Markov chains on trees \cite{ASLS24,SMG24}.

\section{Preliminaries}

\subsection{Cayley tree}

Let $\Gamma^k_+ = (L,E)$ be a semi-infinite Cayley tree of order
$k\geq 1$ with the root $x^0$ (i.e. each vertex of $\Gamma^k_+$ has
exactly $k+1$ edges, except for the root $x^0$, which has $k$
edges). Here $L$ is the set of vertices and $E$ is the set of edges.
The vertices $x$ and $y$ are called {\it nearest neighbors} and they
are denoted by $l=<x,y>$ if there exists an edge connecting them. A
collection of the pairs $<x,x_1>,\dots,<x_{d-1},y>$ is called a {\it
path} from the point $x$ to the point $y$. The distance $d(x,y),
x,y\in V$, on the Cayley tree, is the length of the shortest path
from $x$ to $y$.

Recall a coordinate structure in $\G^k_+$:  every vertex $x$ (except
for $x^0$) of $\G^k_+$ has coordinates $(i_1,\dots,i_n)$, here
$i_m\in\{1,\dots,k\}$, $1\leq m\leq n$ and for the vertex $x^0$ we
put $(0)$.  Namely, the symbol $(0)$ constitutes level 0, and the
sites $(i_1,\dots,i_n)$ form level $n$ (i.e. $d(x^0,x)=n$) of the
lattice (see Fig. 1).

Let us set
\[
W_n = \{ x\in L \, : \, d(x,x_0) = n\} , \qquad \Lambda_n =
\bigcup_{k=0}^n W_k, \qquad  \L_{[n,m]}=\bigcup_{k=n}^mW_k, \ (n<m)
\]
\[
E_n = \big\{ <x,y> \in E \, : \, x,y \in \Lambda_n\big\}, \qquad
\Lambda_n^c = \bigcup_{k=n}^\infty W_k
\]
For $x\in \G^k_+$, $x=(i_1,\dots,i_n)$ denote
$$ S(x)=\{(x,i):\ 1\leq
i\leq k\}.
$$
Here $(x,i)$ means that $(i_1,\dots,i_n,i)$. This set is called a
set of {\it direct successors} of $x$.


\subsection{Quantum Markov Chains}

 The algebra of observables $\cb_x$ for any single site
$x\in L$ will be taken as the algebra $M_d$ of the complex
$d\times d$ matrices. The algebra of observables localized in the
finite volume $\L\subset L$ is then given by
$\cb_\L=\bigotimes\limits_{x\in\L}\cb_x$. As usual if
$\L^1\subset\L^2\subset L$, then $\cb_{\L^1}$ is identified as a
subalgebra of $\cb_{\L^2}$ by tensoring with unit matrices on the
sites $x\in\L^2\setminus\L^1$. Note that, in the sequel, by
$\cb_{\L,+}$ we denote the positive part of $\cb_\L$. The full
algebra $\cb_L$ of the tree is obtained in the usual manner by an
inductive limit
$$
\cb_L=\overline{\bigcup\limits_{\L_n}\cb_{\L_n}}.
$$
In what follows, by ${\cal S}({\cal B}_\L)$ we will denote the set
of all states defined on the algebra ${\cal B}_\L$.

Consider a triplet ${\cal C} \subset {\cal B} \subset {\cal A}$ of
unital $C^*$-algebras. Recall \cite{[Ac74f]} that a {\it
quasi-conditional expectation} with respect to the given triplet
is a completely positive (CP) linear map $\ce \,:\, {\cal A} \to
{\cal B}$ such that $ \ce(ca) = c \ce(a)$, for all $a\in {\cal A},\, c \in {\cal C}$.

\begin{definition}[\cite{AOM}]\label{QMCdef}
A state $\varphi$ on ${\cal B}_L$ is called a {\it (backward) quantum
Markov chain (QMC)}, associated to $\{\L_n\}$, if for each
$\Lambda_n$, there exist a quasi-conditional expectation
$\ce_{\Lambda_n}$ with respect to the triplet
\begin{equation}\label{trplt1}
{\cal B}_{{\Lambda}_{n-1}}\subseteq {\cal
B}_{\Lambda_n}\subseteq{\cal B}_{\Lambda_{n+1}}
\end{equation}
and an initial state $\rho \in S(B_{\Lambda_0})$ such that:
\begin{equation}\label{dfgqmf}
\varphi = \lim_{n\to\infty} \rho_0 \circ
\ce_{\Lambda_0}\circ \ce_{\Lambda_{1}} \circ \cdots \circ
\ce_{\Lambda_n}
\end{equation}
in the weak-* topology.
\end{definition}

\begin{remark}
Notice that in \cite{[AcFiMu07],AMSOa3} a more general definition of backward QMC is given on
arbitrary quasi-local algebras.
\end{remark}

\section{Constructions of QMCs associated with the mixed Ising-$XY$ model on a Cayley tree of order two}

In this section we recall the standard construction of a quantum Markov chain associated with the mixed Ising-$XY$ model on the semi-infinite Cayley tree $\Gamma_{2}^{+}=(L,E)$ of order two.

Let us rewrite the elements of $W_n$ in the lexicographic order with respect to the coordinate system, namely
\begin{eqnarray*}
\overrightarrow{W_n}:=\left(x^{(1)}_{W_n},x^{(2)}_{W_n},\cdots,x^{(|W_n|)}_{W_n}\right).
\end{eqnarray*}
Note that $|W_n|=2^n$. In what follows, by $\prod$ we mean an ordered product, i.e.
$$
\prod_{k=1}^n a_k=a_1a_2\cdots a_n,
$$
where the factors are multiplied in the indicated order.

In the sequel, we consider the quasi-local algebra $B_{L}$ with single-site algebra $B_x=M_{2}(\mathbb{C})$ for every $x\in L$. By $\sigma_{x}^{(u)}$, $\sigma_{y}^{(u)}$, $\sigma_{z}^{(u)}$ one denotes the Pauli spin operators at the site $u\in L$, i.e.
$$
\id^{(u)}=\begin{pmatrix}
1 & 0\\
0& 1
\end{pmatrix},\ \sigma_{x}^{(u)}=\begin{pmatrix}
0 & 1\\
1& 0
\end{pmatrix},\ \sigma_{y}^{(u)}=\begin{pmatrix}
0 & -i\\
i& 0
\end{pmatrix},\ \sigma_{z}^{(u)}=\begin{pmatrix}
1 & 0\\
0& -1
\end{pmatrix}.
$$

For every vertex $u\in L$ we assign the operators
\begin{equation}\label{Kxy}
K_{\langle u,(u,1)\rangle}^{XY}=\exp \left\{J_{XY} \beta H_{\langle u,(u,1)\rangle}^{XY}\right\}, \qquad K_{\langle u,(u,2)\rangle}^{I}=\exp \left\{J_I \beta H_{\langle u,(u,2)\rangle}^{I}\right\},
\end{equation}
where $J_I,J_{XY}\in\mathbb{R}$ and $\beta>0$. Here the corresponding Hamiltonians are defined by
\begin{equation}
\begin{aligned}
& H^{XY}_{\langle u,(u,1)\rangle}=\frac{1}{2}\left(\sigma_x^{(u)} \otimes \sigma_x^{(u,1)}+\sigma_y^{(u)} \otimes \sigma_y^{(u,1)}\right), \\
& H^I_{\langle u,(u,2)\rangle}=\frac{1}{2}\left(\id^{(u)} \otimes \id^{(u,2)}+\sigma_z^{(u)} \otimes \sigma_z^{(u,2)}\right).
\end{aligned}
\end{equation}
Thus the model is mixed in the sense that, at each vertex $u$, one outgoing edge carries an $XY$ interaction and the other outgoing edge carries an Ising interaction.

Now we define a state on $\mathcal{B}_{\Lambda_n}$ with boundary conditions $\omega_{0}\in\mathcal{B}_{(0),+}$ and $\{h_{x}\in\mathcal{B}_{x,+}:\ x\in L\}$.

\begin{figure}[t]
\centering
\begin{tikzpicture}[
  scale=0.92,
  every node/.style={font=\small},
  site/.style={circle,draw,thick,minimum size=6.5mm,inner sep=1pt,fill=white},
  xyedge/.style={thick,decorate,decoration={snake,amplitude=.45mm,segment length=2.2mm}},
  iedge/.style={thick,densely dashed},
  edgelabel/.style={font=\scriptsize,fill=white,inner sep=1pt},
  levellabel/.style={font=\scriptsize}
]
\node[site] (v0) at (0,0) {$0$};
\node[site] (v1) at (-2.8,-1.35) {$1$};
\node[site] (v2) at (2.8,-1.35) {$2$};
\node[site] (v11) at (-4.2,-2.8) {$11$};
\node[site] (v12) at (-1.4,-2.8) {$12$};
\node[site] (v21) at (1.4,-2.8) {$21$};
\node[site] (v22) at (4.2,-2.8) {$22$};

\draw[xyedge] (v0) -- node[edgelabel,above left] {$XY$} (v1);
\draw[iedge] (v0) -- node[edgelabel,above right] {Ising} (v2);
\draw[xyedge] (v1) -- node[edgelabel,above left] {$XY$} (v11);
\draw[iedge] (v1) -- node[edgelabel,above right] {Ising} (v12);
\draw[xyedge] (v2) -- node[edgelabel,above left] {$XY$} (v21);
\draw[iedge] (v2) -- node[edgelabel,above right] {Ising} (v22);

\node[levellabel,anchor=east] at (-5.05,0) {$W_0$};
\node[levellabel,anchor=east] at (-5.05,-1.35) {$W_1$};
\node[levellabel,anchor=east] at (-5.05,-2.8) {$W_2$};

\draw[xyedge] (-4.65,-3.55) -- (-3.85,-3.55);
\node[anchor=west,font=\scriptsize] at (-3.75,-3.55)
  {$XY$ edge: $H_{\langle x,(x,1)\rangle}^{XY}$};
\draw[iedge] (0.35,-3.55) -- (1.15,-3.55);
\node[anchor=west,font=\scriptsize] at (1.25,-3.55)
  {Ising edge: $H_{\langle x,(x,2)\rangle}^{I}$};
\node[font=\scriptsize,align=center,text width=11.4cm] at (0,-4.12)
  {At every vertex $x$, the first successor $(x,1)$ carries the $XY$ interaction, while the second successor $(x,2)$ carries the Ising interaction.};
\end{tikzpicture}
\caption{The first levels of $\G_+^2$ and the locally asymmetric mixed Ising--$XY$ assignment of interactions.}
\label{fig1}
\end{figure}
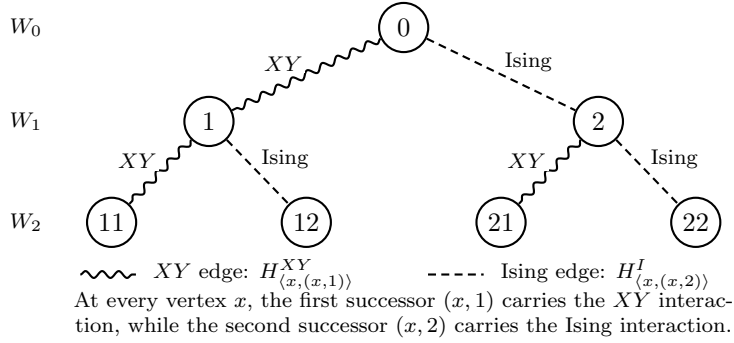

For each $x\in L$ let (see Fig \ref{fig:local-interaction-block})
\begin{equation}\label{K1}
A_{x,(x,1),(x,2)}=K_{\langle x,(x,1)\rangle}^{XY}K_{\langle x,(x,2)\rangle}^{I}.
\end{equation}

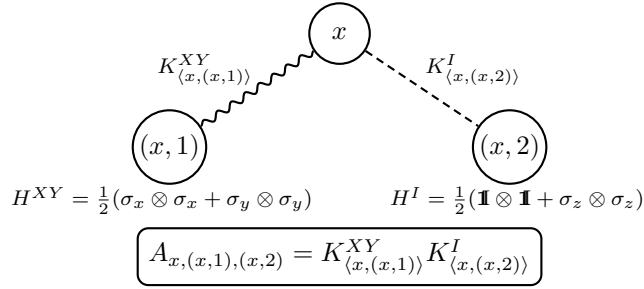
\begin{figure}[t]
\centering
\begin{tikzpicture}[
  every node/.style={font=\small},
  site/.style={circle,draw,thick,minimum size=8mm,inner sep=1pt,fill=white},
  xyedge/.style={thick,decorate,decoration={snake,amplitude=.45mm,segment length=2.2mm}},
  iedge/.style={thick,densely dashed},
  edgelabel/.style={font=\scriptsize,fill=white,inner sep=1pt},
  box/.style={draw,rounded corners,thick,align=center,inner sep=4pt}
]
\node[site] (x) at (0,0) {$x$};
\node[site] (x1) at (-2.25,-1.5) {$(x,1)$};
\node[site] (x2) at (2.25,-1.5) {$(x,2)$};

\draw[xyedge] (x) -- node[edgelabel,above left] {$K_{\langle x,(x,1)\rangle}^{XY}$} (x1);
\draw[iedge] (x) -- node[edgelabel,above right] {$K_{\langle x,(x,2)\rangle}^{I}$} (x2);

\node[box] at (0,-2.95)
  {$A_{x,(x,1),(x,2)}=K_{\langle x,(x,1)\rangle}^{XY}K_{\langle x,(x,2)\rangle}^{I}$};
\node[align=center,font=\scriptsize] at (-2.35,-2.2)
  {$H^{XY}=\frac12(\sigma_x\otimes\sigma_x+\sigma_y\otimes\sigma_y)$};
\node[align=center,font=\scriptsize] at (2.35,-2.2)
  {$H^I=\frac12(\id\otimes\id+\sigma_z\otimes\sigma_z)$};
\end{tikzpicture}
\caption{The local three-site interaction block at a vertex $x$. The left child edge is the coherent $XY$ exchange edge, while the right child edge is the diagonal Ising edge.}
\label{fig:local-interaction-block}
\end{figure}

Using these operators, we put
\begin{eqnarray}
\label{K11} &&K_{[m,m+1]}:=\prod_{x\in \overrightarrow{W_{m}}}A_{x,(x,1),(x,2)}, \qquad 0\leq m\leq n-1,\\[2mm]
\label{K2} &&
\bh_{n}^{1/2}:=\prod_{x\in\overrightarrow{W_n}}(h_{x})^{1/2}, \qquad \bh_n=\bh_{n}^{1/2}(\bh_{n}^{1/2})^*,\\[2mm]
\label{K3}
&&{\mathbf{K}}_n:=\omega_{0}^{1/2}\prod_{m=0}^{n-1}K_{[m,m+1]}\bh_{n}^{1/2},\\[2mm]
\label{K4} &&
\mathcal{W}_{n]}:={\mathbf{K}}_n^{*}{\mathbf{K}}_n.
\end{eqnarray}

One can see that $\mathcal{W}_{n]}$ is positive. In what follows, by $\tr_{\Lambda}:\cb_L\to\cb_{\Lambda}$ we mean the normalized partial trace (i.e. $\tr_{\Lambda}(\id_{L})=\id_{\Lambda}$ for every finite $\Lambda\subset L$). For brevity we write $\tr_{n]}:=\tr_{\Lambda_n}$. We also denote by $\Tr$ the normalized trace on the single-site algebra $M_2(\mathbb C)$, so that $\Tr(\id)=1$. For $x\in L$, we also denote by $\tr_{S(x)}$ the normalized partial trace over the tensor factors corresponding to the direct successors of $x$.

For later reference we call
\[
\mathcal E_x(Y):=\tr_{S(x)}\bigl(A_{x,(x,1),(x,2)}Y A_{x,(x,1),(x,2)}^*\bigr),
\qquad
Y\in\mathcal B_x\otimes\mathcal B_{(x,1)}\otimes\mathcal B_{(x,2)},
\]
the local transfer operator at the vertex $x$. Thus the compatibility relation for the boundary law is
\[
h_x=\mathcal E_x\bigl(\id^{(x)}\otimes h_{(x,1)}\otimes h_{(x,2)}\bigr),
\]
which is exactly the local equation appearing in \eqref{eq2}.

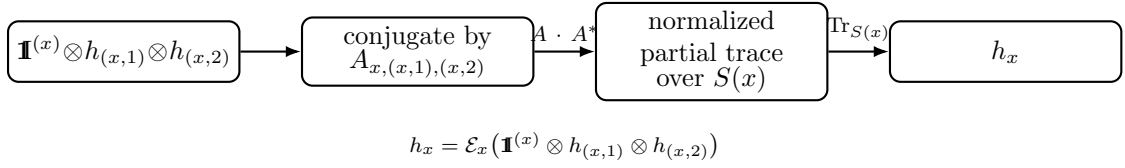
\begin{figure}[t]
\centering
\begin{tikzpicture}[
  node distance=8mm,
  every node/.style={font=\small},
  box/.style={draw,rounded corners,thick,align=center,minimum height=8mm,text width=2.85cm,inner sep=3pt},
  arr/.style={-{Latex[length=2mm]},thick}
]
\node[box] (input) {$\id^{(x)}\otimes h_{(x,1)}\otimes h_{(x,2)}$};
\node[box,right=of input] (conj) {conjugate by\\[-1mm]$A_{x,(x,1),(x,2)}$};
\node[box,right=of conj] (trace) {normalized partial trace\\[-1mm]over $S(x)$};
\node[box,right=of trace] (output) {$h_x$};

\draw[arr] (input) -- (conj);
\draw[arr] (conj) -- node[above,font=\scriptsize] {$A\,\cdot\,A^*$} (trace);
\draw[arr] (trace) -- node[above,font=\scriptsize] {$\tr_{S(x)}$} (output);

\node[align=center,font=\scriptsize] at ($(input)!0.5!(output)+(0,-1.25)$)
  {$h_x=\mathcal E_x\bigl(\id^{(x)}\otimes h_{(x,1)}\otimes h_{(x,2)}\bigr)$};
\end{tikzpicture}
\caption{The local transfer operator as a one-step boundary-law update: successor boundary data are inserted at the two children, conjugated by the local interaction block, and traced over the successors to produce the parent boundary operator.}
\label{fig:local-transfer-diagram}
\end{figure}

Define the positive functional $\ffi^{(n)}_{\omega_0,\bh}$ by
\begin{eqnarray}\label{ffi-ff}
\ffi^{(n)}_{\omega_0,\bh}(a)=\tr(\cw_{n+1]}(a\otimes\id_{W_{n+1}})),
\end{eqnarray}
for every $a\in \cb_{\Lambda_n}$.

To get an infinite-volume state $\ffi$ on $\cb_L$ such that
$\ffi\lceil_{\cb_{\Lambda_n}}=\ffi^{(n)}_{\omega_0,\bh}$, we impose compatibility conditions on the boundary data $\{\omega_0,\bh\}$, namely
\begin{eqnarray}\label{compatibility}
\ffi^{(n+1)}_{\omega_0,\bh}\lceil_{\cb_{\Lambda_n}}=\ffi^{(n)}_{\omega_0,\bh}.
\end{eqnarray}

\begin{theorem}\label{compa}\cite{AMSOa3}
Assume that for every $x\in L$ the operators $K_{\langle x,(x,1)\rangle}^{XY}$ and $K_{\langle x,(x,2)\rangle}^{I}$ are given as above. Let the boundary conditions $\omega_{0}\in {\cal B}_{(0),+}$ and ${\bh}=\{h_x\in {\cal B}_{x,+}\}_{x\in L}$ satisfy (see Fig. \ref{fig:local-transfer-diagram})
\begin{eqnarray}\label{eq1}
&&\Tr(\omega_{0}h_{0})=1, \\
\label{eq2} &&\tr_{S(x)}\big({A_{x,(x,1),(x,2)}}(\id\otimes h_{(x,1)}\otimes h_{(x,2)})A_{x,(x,1),(x,2)}^{*}\big)=h_x, \qquad \textrm{for every } x\in L.
\end{eqnarray}
Then the functionals $\{\ffi^{(n)}_{\omega_0,\bh}\}$ satisfy the compatibility condition \eqref{compatibility}. Moreover, there is a unique quantum Markov chain $\ffi_{\omega_0,{\bh}}$ on $\cb_L$ such that
$$
\ffi_{\omega_0,{\bh}}=w-\lim\limits_{n\to\infty}\ffi^{(n)}_{\omega_0,\bh}.
$$
\end{theorem}

\section{Translation-invariant boundary equation}

In this section we solve the translation-invariant version of \eqref{eq2}. We therefore assume that $h_x=\mathbf h$ for every $x\in L$. Since $\Tr$ is the normalized one-site trace, we write $\mathbf h$ in the Pauli basis as
\begin{equation}\label{bloch}
\mathbf h=X\id+x\sigma_x+y\sigma_y+z\sigma_z
\end{equation}
so that
\[
X=\Tr(\mathbf h),\qquad x=\Tr(\mathbf h\sigma_x),\qquad y=\Tr(\mathbf h\sigma_y),\qquad z=\Tr(\mathbf h\sigma_z).
\]
Since $\mathbf h\ge 0$, one has
\begin{equation}\label{positivity}
X>0, \qquad x^2+y^2+z^2\le X^2.
\end{equation}

Let us first calculate the operators $K_{\langle u,(u,1)\rangle}^{XY}$ and $K_{\langle u,(u,2)\rangle}^{I}$. We obtain
\begin{equation}\label{KI}
K_{\langle u,(u,2)\rangle}^{I}=\kappa_0 \id^{(u)} \otimes \id^{(u,2)}+\kappa_1 \sigma_z^{(u)} \otimes \sigma_z^{(u,2)},
\end{equation}
\begin{equation}\label{KXY}
K_{\langle u,(u,1)\rangle}^{XY}=\frac{C+1}{2}\,\id^{(u)}\otimes\id^{(u,1)}-\frac{C-1}{2}\,\sigma_z^{(u)}\otimes\sigma_z^{(u,1)}
+\frac{S}{2}\left(\sigma_x^{(u)}\otimes\sigma_x^{(u,1)}+\sigma_y^{(u)}\otimes\sigma_y^{(u,1)}\right),
\end{equation}
where
\begin{equation}\label{KK}
\kappa_0=\frac{1}{2}(e^{J_I \beta}+1), \qquad \kappa_1=\frac{1}{2}(e^{J_I \beta}-1),
\end{equation}
\begin{equation}\label{SC}
S=\sinh (J_{XY}\beta), \qquad C=\cosh (J_{XY}\beta).
\end{equation}

Multiplying the operators in the order prescribed by \eqref{K1}, one gets
\begin{equation}\label{AI}
\begin{aligned}
A_{u,(u,1),(u,2)}&=\frac{C+1}{2}\kappa_0\,\id^{(u)} \otimes \id^{(u,1)} \otimes \id^{(u,2)}
+\frac{C+1}{2}\kappa_1\,\sigma_z^{(u)}\otimes\id^{(u,1)} \otimes \sigma_z^{(u,2)}\\[2mm]
&\quad - \frac{C-1}{2}\kappa_1\,\id^{(u)}\otimes\sigma_z^{(u,1)}\otimes \sigma_z^{(u,2)}
- \frac{C-1}{2}\kappa_0\,\sigma_z^{(u)}\otimes\sigma_z^{(u,1)}\otimes \id^{(u,2)}\\[2mm]
&\quad + \frac{S\kappa_0}{2}\left(\sigma_x^{(u)} \otimes \sigma_x^{(u,1)}\otimes \id^{(u,2)}+\sigma_y^{(u)} \otimes \sigma_y^{(u,1)}\otimes \id^{(u,2)}\right)\\[2mm]
&\quad +\frac{S\kappa_1}{2}\left(i\sigma_x^{(u)} \otimes \sigma_y^{(u,1)}\otimes \sigma_z^{(u,2)}-i\sigma_y^{(u)} \otimes \sigma_x^{(u,1)}\otimes \sigma_z^{(u,2)}\right).
\end{aligned}
\end{equation}

For convenience, put
\begin{equation}\label{abcd}
a=C^2(\kappa_0^2+\kappa_1^2), \qquad b=2S^2\kappa_0\kappa_1, \qquad c=S(\kappa_0^2+\kappa_1^2), \qquad d=2\kappa_0\kappa_1.
\end{equation}

\begin{proposition}\label{prop:reduced}
For a translation-invariant boundary condition $h_x=\mathbf h$, the equation \eqref{eq2} is equivalent to
\begin{equation}\label{Phih}
\tr_{S(u)}\bigg( A_{u,(u,1),(u,2)}\left[\id^{(u)} \otimes \mathbf h\otimes\mathbf h\right] A_{u,(u,1),(u,2)}^{*}\bigg)
= (aX^2-bz^2)\id+cXx\,\sigma_x+cXy\,\sigma_y+dXz\,\sigma_z,
\end{equation}
where $X,x,y,z$ are given by \eqref{bloch}. Consequently, \eqref{eq2} reduces to the system
\begin{equation}\label{TH}
\left\{
\begin{aligned}
X&=aX^2-bz^2,\\[1mm]
x&=cXx,\\[1mm]
y&=cXy,\\[1mm]
z&=dXz.
\end{aligned}
\right.
\end{equation}
\end{proposition}

\begin{proof}
This follows from the detailed computation given in Appendix \ref{app:reducedmap}.
\end{proof}

\begin{theorem}\label{Main}
Let $J_I\in\mathbb R$, $J_{XY}\in\mathbb R$, and $\beta>0$. Then the system \eqref{TH} has a unique positive solution, namely
\begin{equation}\label{hstar}
X_*=\frac{1}{C^2(\kappa_0^2+\kappa_1^2)}, \qquad x_*=y_*=z_*=0,
\end{equation}
or equivalently,
\begin{equation}\label{hstar-matrix}
\mathbf h_*=
\frac{1}{C^2(\kappa_0^2+\kappa_1^2)}\,\id.
\end{equation}
\end{theorem}

\begin{proof}
Assume first that $z\neq 0$. Then the last equation in \eqref{TH} yields
\begin{equation}\label{Xeqd}
X=\frac{1}{d}.
\end{equation}
Since $X=\Tr(\mathbf h)>0$, relation \eqref{Xeqd} is possible only when $d>0$, equivalently $J_I>0$.

If $S=0$, then $b=0$, and the first equation in \eqref{TH} becomes
\[
\frac{1}{d}=a\frac{1}{d^2},
\]
which implies $a=d$. However, by \eqref{abcd} and \eqref{SC},
\[
a-d=\kappa_0^2+\kappa_1^2-2\kappa_0\kappa_1=(\kappa_0-\kappa_1)^2=1>0,
\]
a contradiction. Therefore $z\neq 0$ is impossible when $S=0$.

Assume now that $S\neq 0$. Combining \eqref{TH} with \eqref{Xeqd}, we obtain
\begin{equation}\label{zvalue}
z^2=\frac{a-d}{bd^2}.
\end{equation}
Since $\mathbf h\ge 0$, condition \eqref{positivity} gives $0\le X^2-z^2$. Using \eqref{Xeqd}, \eqref{zvalue}, and \eqref{abcd}, we compute
\begin{align*}
X^2-z^2
&=\frac{1}{d^2}-\frac{a-d}{bd^2}
=\frac{b-a+d}{bd^2}\\[1mm]
&=\frac{2S^2\kappa_0\kappa_1-C^2(\kappa_0^2+\kappa_1^2)+2\kappa_0\kappa_1}{bd^2}
=\frac{2C^2\kappa_0\kappa_1-C^2(\kappa_0^2+\kappa_1^2)}{bd^2}\\[1mm]
&=-\frac{C^2(\kappa_0-\kappa_1)^2}{bd^2}<0,
\end{align*}
which contradicts \eqref{positivity}. Hence necessarily
\begin{equation}\label{zzero}
z=0.
\end{equation}

With \eqref{zzero}, the first equation in \eqref{TH} reduces to $X=aX^2$. Since $X>0$, we obtain
\[
X=\frac{1}{a}=\frac{1}{C^2(\kappa_0^2+\kappa_1^2)}.
\]
Substituting this into the second and third equations of \eqref{TH}, we find
\[
x=\frac{c}{a}x=\frac{S}{C^2}x, \qquad y=\frac{c}{a}y=\frac{S}{C^2}y.
\]
But
\[
C^2-S=1+S^2-S=\left(S-\frac12\right)^2+\frac34>0,
\]
so $S/C^2\neq 1$. Therefore $x=y=0$. This proves \eqref{hstar}--\eqref{hstar-matrix}.
\end{proof}

\begin{corollary}\label{cor:qmc}
Let
\begin{equation}\label{omega0}
\omega_0=C^2(\kappa_0^2+\kappa_1^2)\,\id,
\end{equation}
and let $h_x=\mathbf h_*$ for every $x\in L$, where $\mathbf h_*$ is given by \eqref{hstar-matrix}. Then \eqref{eq1}--\eqref{eq2} hold, and Theorem \ref{compa} yields a translation-invariant quantum Markov chain $\ffi_*$. Moreover, $\ffi_*$ is the unique QMC generated by a positive translation-invariant boundary condition.
\end{corollary}

\begin{proof}
By construction,
\[
\Tr(\omega_0\mathbf h_*)=\Tr\left(C^2(\kappa_0^2+\kappa_1^2)\,\id\cdot \frac{1}{C^2(\kappa_0^2+\kappa_1^2)}\,\id\right)=1.
\]
Theorem \ref{Main} gives the uniqueness of the translation-invariant positive solution of \eqref{eq2}, and Theorem \ref{compa} completes the construction.
\end{proof}

The system \eqref{TH} allows us to define a mapping \(
\mathcal F:\mathbb R^4\to \mathbb R^4\) by
\begin{equation}\label{mapF}
\mathcal F(X,x,y,z)=\bigl(aX^2-bz^2,\; cXx,\; cXy,\; dXz\bigr),
\end{equation}
where, as before,
\[
a=C^2(\kappa_0^2+\kappa_1^2),\qquad
b=2S^2\kappa_0\kappa_1,\qquad
c=S(\kappa_0^2+\kappa_1^2),\qquad
d=2\kappa_0\kappa_1.
\]
The map $\mathcal F$ is the reduced boundary-law dynamics in the translation-invariant sector: it is obtained by applying the local transfer operator to a boundary matrix $\mathbf h=X\id+x\sigma_x+y\sigma_y+z\sigma_z$ and reading off the resulting Bloch coefficients. In the estimates below we shall also use the normalized ratios $x/X$, $y/X$, and $z/X$; these ratios remove the scalar part of the positive boundary operator while keeping the positivity constraint explicit.

Let
\[
\mathcal P:=\Bigl\{(X,x,y,z)\in\mathbb R^4:\ X>0,\ x^2+y^2+z^2\le X^2\Bigr\}
\]
be the admissible region.

A periodic point $P\in\mathbb R^4$ of $\mathcal F$ is called \emph{admissible} if, for every $n\ge0$,
\[
\mathcal F^n(P)=(X_n,x_n,y_n,z_n)
\]
belongs to $\mathcal P$.

\begin{theorem}\label{thm:no-admissible-periodic}
Every admissible periodic point of $\mathcal F$ is a fixed point. More precisely, its whole orbit is the constant orbit
\[
\left(\frac1a,0,0,0\right).
\]
In particular, the system \eqref{TH} has no admissible periodic points with minimal period greater than one.
\end{theorem}

\begin{proof}
Let $P$ be an admissible periodic point of $\mathcal F$, and write
\[
\mathcal F^n(P)=(X_n,x_n,y_n,z_n),\qquad n\ge0.
\]
Since the orbit is admissible, we have
\[
X_n>0,
\qquad
x_n^2+y_n^2+z_n^2\le X_n^2,
\qquad n\ge0.
\]
In particular,
\[
|z_n|\le X_n.
\]

Define
\[
u_n:=\frac{z_n}{X_n},\qquad n\ge0.
\]
Then $|u_n|\le 1$ for all $n$. From the system \eqref{TH} we obtain
\[
X_{n+1}=aX_n^2-bz_n^2=X_n^2\bigl(a-bu_n^2\bigr),
\]
and
\[
z_{n+1}=dX_n z_n=dX_n^2u_n.
\]
Hence
\begin{equation}\label{eq:u-recursion}
u_{n+1}=\frac{d\,u_n}{a-bu_n^2}.
\end{equation}

We now show that
\begin{equation}\label{eq:key-denominator}
a-bu^2>|d|
\qquad\text{for every }u\in[-1,1].
\end{equation}

If $d\ge0$, then $b=S^2d\ge0$, and therefore
\[
a-bu^2\ge a-b=C^2(\kappa_0^2+\kappa_1^2)-2S^2\kappa_0\kappa_1=d+C^2>|d|.
\]

Assume next that $d<0$. Then $b<0$, so $a-bu^2\ge a$. Moreover,
\[
a-|d|=a+d=C^2(\kappa_0^2+\kappa_1^2)+2\kappa_0\kappa_1.
\]
Since
\[
d=2\kappa_0\kappa_1=\frac{e^{2\beta J_I}-1}{2}>-\frac12,
\]
we have
\[
a+d>C^2\Bigl(-\frac12+1\Bigr)-\frac12=\frac{C^2-1}{2}=\frac{S^2}{2}\ge0.
\]
Thus $a>|d|$, and again \eqref{eq:key-denominator} holds.

Now \eqref{eq:u-recursion} and \eqref{eq:key-denominator} imply that whenever $u_n\neq0$,
\[
|u_{n+1}|=\frac{|d|}{a-bu_n^2}\,|u_n|<|u_n|.
\]
Since $(u_n)$ is periodic, this is impossible unless $u_n=0$ for every $n\ge0$. Hence
\[
z_n=0\qquad\text{for all }n\ge0.
\]

The first equation of \eqref{TH} now reduces to
\[
X_{n+1}=aX_n^2.
\]
Set $v_n:=aX_n$. Then $v_n>0$ and
\[
v_{n+1}=v_n^2.
\]
If the orbit has period $m$, then
\[
v_n=v_{n+m}=v_n^{2^m},
\]
so $v_n^{2^m-1}=1$. Because $v_n>0$, it follows that $v_n=1$ for all $n\ge0$. Therefore
\[
X_n=\frac1a\qquad\text{for all }n\ge0.
\]

Using this in the second and third equations of \eqref{TH}, we obtain
\[
x_{n+1}=\frac{c}{a}x_n,
\qquad
y_{n+1}=\frac{c}{a}y_n.
\]
Since
\[
\left|\frac{c}{a}\right|=\frac{|S|}{C^2}<1,
\]
periodicity forces $x_n=y_n=0$ for all $n\ge0$.

We conclude that every point of the orbit is equal to
\[
\left(\frac1a,0,0,0\right).
\]
Hence every admissible periodic point is a fixed point, and there are no admissible periodic points with minimal period greater than one.
\end{proof}

\section{Infinite admissible inverse trajectories and their consequence for QMCs}
\label{subsec:inverse-trajectories}

Let \(
\mathcal F:\mathbb R^4\to\mathbb R^4\) be the mapping given by \eqref{mapF}.

\begin{theorem}\label{thm:infinite-backward}
Assume that $J_{XY} J_I\neq 0$. Let
\[
P_n=(X_n,x_n,y_n,z_n)\in\mathcal P,
\qquad n\ge 0,
\]
be an infinite admissible inverse trajectory of \(\mathcal F\), i.e.
\[
\mathcal F(P_{n+1})=P_n,
\qquad n\ge 0.
\]
Then
\[
x_n=y_n=z_n=0
\qquad\text{for all }n\ge 0,
\]
and
\[
X_n=\frac{(aX_0)^{2^{-n}}}{a},
\qquad n\ge 0.
\]
Conversely, for every \(A>0\), the sequence
\[
P_n^{(A)}
=
\left(\frac{(aA)^{2^{-n}}}{a},\,0,\,0,\,0\right),
\qquad n\ge 0,
\]
is an infinite admissible inverse trajectory of \(\mathcal F\).
\end{theorem}

\begin{proof}
Let
\[
P_n=(X_n,x_n,y_n,z_n)\in\mathcal P,
\qquad
\mathcal F(P_{n+1})=P_n.
\]
Since \(X_n>0\), define the normalized variables
\[
u_n:=\frac{x_n}{X_n},
\qquad
v_n:=\frac{y_n}{X_n},
\qquad
w_n:=\frac{z_n}{X_n}.
\]
Then
\[
u_n^2+v_n^2+w_n^2\le 1,
\qquad n\ge 0.
\]
From \(\mathcal F(P_{n+1})=P_n\) we obtain
\begin{equation}\label{eq:backward-ratios}
\begin{aligned}
X_n&=X_{n+1}^2\bigl(a-bw_{n+1}^2\bigr),\\
u_n&=\frac{c}{a-bw_{n+1}^2}\,u_{n+1},\\
v_n&=\frac{c}{a-bw_{n+1}^2}\,v_{n+1},\\
w_n&=\frac{d}{a-bw_{n+1}^2}\,w_{n+1}.
\end{aligned}
\end{equation}

We first show that
\begin{equation}\label{eq:key-ineq}
a-bt^2>|d|
\qquad\text{for every }t\in[-1,1].
\end{equation}
Indeed, if \(d>0\), then \(b=S^2d>0\), hence
\[
a-bt^2\ge a-b
=
C^2(\kappa_0^2+\kappa_1^2)-2S^2\kappa_0\kappa_1
=
C^2+d
>
d=|d|.
\]
If \(d<0\), then \(b<0\), so
\[
a-bt^2\ge a=C^2(\kappa_0^2+\kappa_1^2).
\]
Using
\[
\kappa_0^2+\kappa_1^2=d+1,
\qquad d>-\frac12,
\]
we get
\[
a-|d|
=
C^2(d+1)+d
=
d(C^2+1)+C^2
>
-\frac12(C^2+1)+C^2
=
\frac{S^2}{2}>0,
\]
because \(J_{XY}\neq 0\) implies \(S\neq 0\). This proves
\eqref{eq:key-ineq}.

By compactness of \([-1,1]\), \eqref{eq:key-ineq} implies that
\[
\theta:=\max_{|t|\le 1}\frac{|d|}{a-bt^2}
\]
is well-defined and satisfies \(0<\theta<1\). Therefore, from
\eqref{eq:backward-ratios},
\[
|w_n|
\le
\theta\,|w_{n+1}|,
\qquad n\ge 0.
\]
Iterating,
\[
|w_0|
\le
\theta^N |w_N|
\le
\theta^N,
\qquad N\ge 1.
\]
Since \(\theta<1\), letting \(N\to\infty\) gives \(w_0=0\). Applying the same
argument with the shifted sequence \((P_{n+m})_{n\ge 0}\), we obtain
\[
w_n=0
\qquad\text{for all }n\ge 0.
\]
Hence
\[
z_n=0
\qquad\text{for all }n\ge 0.
\]

With \(w_n=0\), the second and third equations in \eqref{eq:backward-ratios}
become
\[
u_n=\frac{c}{a}u_{n+1},
\qquad
v_n=\frac{c}{a}v_{n+1}.
\]
Now
\[
\left|\frac{c}{a}\right|
=
\frac{|S|}{C^2}
<1.
\]
Therefore
\[
|u_0|
\le
\left|\frac{c}{a}\right|^N |u_N|
\le
\left|\frac{c}{a}\right|^N,
\qquad
|v_0|
\le
\left|\frac{c}{a}\right|^N |v_N|
\le
\left|\frac{c}{a}\right|^N.
\]
Letting \(N\to\infty\) yields \(u_0=v_0=0\). Again shifting the sequence gives
\[
u_n=v_n=0
\qquad\text{for all }n\ge 0.
\]
Thus
\[
x_n=y_n=0
\qquad\text{for all }n\ge 0.
\]

It remains to solve the first equation in \eqref{eq:backward-ratios}, which now
reduces to
\[
X_n=aX_{n+1}^2,
\qquad n\ge 0.
\]
Hence
\[
X_{n+1}=\sqrt{\frac{X_n}{a}},
\qquad n\ge 0.
\]
By induction,
\[
X_n=\frac{(aX_0)^{2^{-n}}}{a},
\qquad n\ge 0.
\]

Conversely, for every \(A>0\), define
\[
X_n=\frac{(aA)^{2^{-n}}}{a},
\qquad x_n=y_n=z_n=0.
\]
Then \(X_n>0\), so \(P_n^{(A)}\in\mathcal P\), and
\[
aX_{n+1}^2
=
a\left(\frac{(aA)^{2^{-(n+1)}}}{a}\right)^2
=
\frac{(aA)^{2^{-n}}}{a}
=
X_n.
\]
Therefore \(\mathcal F(P_{n+1}^{(A)})=P_n^{(A)}\). The proof is complete.
\end{proof}

\begin{corollary}\label{cor:radial-scalar-boundary}
For every \(A>0\), define
\[
t_n^{(A)}:=\frac{(aA)^{2^{-n}}}{a},
\qquad n\ge 0,
\]
and set
\[
h_x^{(A)}:=t_{|x|}^{(A)}\,\id,
\qquad x\in L,
\]
where \(|x|\) denotes the generation of \(x\).
Then
\[
h_x^{(A)}
=
\tr_{(x,1),(x,2)}
\Bigl(
A_{x,(x,1),(x,2)}
\bigl(\id\otimes h_{(x,1)}^{(A)}\otimes h_{(x,2)}^{(A)}\bigr)
A_{x,(x,1),(x,2)}^*
\Bigr),
\]
so \(\{h_x^{(A)}\}_{x\in L}\) is a positive level-dependent boundary law.
If \(A\neq a^{-1}\), then this boundary law is not translation-invariant.
\end{corollary}

\begin{proof}
Since
\[
\tr_{(x,1),(x,2)}\bigl(A_{x,(x,1),(x,2)}A_{x,(x,1),(x,2)}^*\bigr)=a\,\id,
\]
we obtain, for scalar matrices,
\[
\tr_{(x,1),(x,2)}
\Bigl(
A_{x,(x,1),(x,2)}
\bigl(\id\otimes s\id\otimes t\id\bigr)
A_{x,(x,1),(x,2)}^*
\Bigr)
=
ast\,\id.
\]
Taking \(s=t=t_{n+1}^{(A)}\) and using
\[
t_n^{(A)}=a\bigl(t_{n+1}^{(A)}\bigr)^2
\]
gives the claim.
\end{proof}

\begin{proposition}\label{prop:scalar-gauge}
Let \(\{t_x\}_{x\in L}\) be any family of positive numbers satisfying
\begin{equation}\label{eq:scalar-recursion}
t_x=a\,t_{(x,1)}t_{(x,2)},
\qquad x\in L.
\end{equation}
Set
\[
h_x=t_x\,\id,
\qquad
\omega_0=t_0^{-1}\id.
\]
Then the corresponding finite-volume operators \(\mathcal W_{n]}\) are
independent of the choice of \(\{t_x\}\). More precisely,
\[
\mathcal W_{n]}
=
a^{-(2^n-1)}
\Bigl(K_{[0,1]}K_{[1,2]}\cdots K_{[n-1,n]}\Bigr)^*
\Bigl(K_{[0,1]}K_{[1,2]}\cdots K_{[n-1,n]}\Bigr).
\]
Consequently, every scalar boundary law satisfying \eqref{eq:scalar-recursion}
defines the same quantum Markov chain, namely the translation-invariant state
\(\varphi_*\).
\end{proposition}

\begin{proof}
Since each \(h_x\) is scalar,
\[
\mathbf h_n
=
\prod_{x\in W_n} h_x
=
\left(\prod_{x\in W_n} t_x\right)\id.
\]
Hence
\[
\mathbf K_n
=
\omega_0^{1/2}
K_{[0,1]}K_{[1,2]}\cdots K_{[n-1,n]}
\mathbf h_n^{1/2}
=
t_0^{-1/2}
\left(\prod_{x\in W_n} t_x\right)^{1/2}
K_{[0,1]}K_{[1,2]}\cdots K_{[n-1,n]}.
\]
Therefore
\begin{equation}\label{eq:Wn-scalar}
\mathcal W_{n]}
=
t_0^{-1}\left(\prod_{x\in W_n} t_x\right)
\Bigl(K_{[0,1]}K_{[1,2]}\cdots K_{[n-1,n]}\Bigr)^*
\Bigl(K_{[0,1]}K_{[1,2]}\cdots K_{[n-1,n]}\Bigr).
\end{equation}

Now multiply \eqref{eq:scalar-recursion} over all vertices \(x\in W_m\):
\[
\prod_{x\in W_m} t_x
=
a^{|W_m|}\prod_{y\in W_{m+1}} t_y.
\]
Iterating this identity from \(m=0\) to \(m=n-1\), we obtain
\[
t_0
=
a^{1+2+\cdots+2^{n-1}}
\prod_{x\in W_n} t_x
=
a^{2^n-1}\prod_{x\in W_n} t_x.
\]
Substituting this into \eqref{eq:Wn-scalar} gives
\[
\mathcal W_{n]}
=
a^{-(2^n-1)}
\Bigl(K_{[0,1]}K_{[1,2]}\cdots K_{[n-1,n]}\Bigr)^*
\Bigl(K_{[0,1]}K_{[1,2]}\cdots K_{[n-1,n]}\Bigr),
\]
which is independent of the particular scalar solution \(\{t_x\}\).

In particular, taking \(t_x\equiv a^{-1}\) gives exactly the finite-volume
operators of the translation-invariant state \(\varphi_*\). Therefore every
scalar boundary law satisfying \eqref{eq:scalar-recursion} defines the same
quantum Markov chain \(\varphi_*\).
\end{proof}

\section{The full two-child boundary equation and rigidity of positive solutions}
\label{subsec:full-two-child}

We now solve the full boundary equation \eqref{eq2} without assuming
translation invariance. For each $x\in L$, write
\begin{equation}\label{eq:hx-full-bloch}
h_x=X_x\id+\xi_x\sigma_x+\eta_x\sigma_y+\zeta_x\sigma_z,
\end{equation}
where, because $\Tr$ is the normalized one-site trace,
\[
X_x=\Tr(h_x),\qquad
\xi_x=\Tr(h_x\sigma_x),\qquad
\eta_x=\Tr(h_x\sigma_y),\qquad
\zeta_x=\Tr(h_x\sigma_z).
\]
Since $h_x\ge 0$, one has
\begin{equation}\label{eq:positivity-full}
X_x>0,
\qquad
\xi_x^2+\eta_x^2+\zeta_x^2\le X_x^2.
\end{equation}

\begin{proposition}\label{prop:full-boundary-map}
Let
\[
h_{(x,1)}=X_1\id+\xi_1\sigma_x+\eta_1\sigma_y+\zeta_1\sigma_z,
\qquad
h_{(x,2)}=X_2\id+\xi_2\sigma_x+\eta_2\sigma_y+\zeta_2\sigma_z.
\]
Then, under the normalized partial trace convention,
\[
\tr_{S(x)}\Bigl(
A_{x,(x,1),(x,2)}
\bigl(\id\otimes h_{(x,1)}\otimes h_{(x,2)}\bigr)
A_{x,(x,1),(x,2)}^*
\Bigr)
\]
is equal to
\begin{equation}\label{eq:full-map}
\bigl(aX_1X_2-b\zeta_1\zeta_2\bigr)\id
+cX_2\xi_1\,\sigma_x
+cX_2\eta_1\,\sigma_y
+dX_1\zeta_2\,\sigma_z,
\end{equation}
where
\[
a=C^2(\kappa_0^2+\kappa_1^2),\qquad
b=2S^2\kappa_0\kappa_1,\qquad
c=S(\kappa_0^2+\kappa_1^2),\qquad
d=2\kappa_0\kappa_1.
\]
Consequently, the full boundary equation \eqref{eq2} is equivalent to
\begin{equation}\label{eq:full-recursion}
\left\{
\begin{aligned}
X_x&=aX_{(x,1)}X_{(x,2)}-b\zeta_{(x,1)}\zeta_{(x,2)},\\[1mm]
\xi_x&=cX_{(x,2)}\xi_{(x,1)},\\[1mm]
\eta_x&=cX_{(x,2)}\eta_{(x,1)},\\[1mm]
\zeta_x&=dX_{(x,1)}\zeta_{(x,2)}.
\end{aligned}
\right.
\end{equation}
\end{proposition}

\begin{proof}
This is a direct expansion of
\[
A_{x,(x,1),(x,2)}
\bigl(\id\otimes h_{(x,1)}\otimes h_{(x,2)}\bigr)
A_{x,(x,1),(x,2)}^*
\]
using \eqref{AI}, followed by the normalized partial trace over the two
successors. Since
\[
\Tr(\sigma_x)=\Tr(\sigma_y)=\Tr(\sigma_z)=0,\qquad \Tr(\id)=1,
\]
all mixed Pauli terms disappear under the partial trace except the ones shown in
\eqref{eq:full-map}. Comparing coefficients in the Pauli basis yields
\eqref{eq:full-recursion}.
\end{proof}

\begin{theorem}\label{thm:rigidity-full-boundary}
Let $J_I\in\mathbb R$, $J_{XY}\in\mathbb R$, and $\beta>0$. Assume that
$\{h_x\}_{x\in L}$ is a family of positive matrices satisfying the full boundary
equation \eqref{eq2}. Then
\[
\xi_x=\eta_x=\zeta_x=0
\qquad\text{for every }x\in L.
\]
Equivalently, every positive solution of \eqref{eq2} is scalar:
\begin{equation}\label{eq:scalar-form}
h_x=t_x\,\id,\qquad t_x>0.
\end{equation}
Moreover, the scalar coefficients satisfy
\begin{equation}\label{eq:scalar-recursion-full}
t_x=a\,t_{(x,1)}t_{(x,2)},
\qquad x\in L.
\end{equation}
\end{theorem}

\begin{proof}
Define the normalized ratios
\[
p_x:=\frac{\xi_x}{X_x},
\qquad
q_x:=\frac{\eta_x}{X_x},
\qquad
r_x:=\frac{\zeta_x}{X_x}.
\]
By \eqref{eq:positivity-full},
\[
p_x^2+q_x^2+r_x^2\le 1,
\qquad x\in L.
\]
In particular,
\[
|p_x|\le 1,\qquad |q_x|\le 1,\qquad |r_x|\le 1.
\]

Dividing \eqref{eq:full-recursion} by $X_x>0$, we obtain
\begin{equation}\label{eq:ratio-recursion}
\left\{
\begin{aligned}
X_x&=X_{(x,1)}X_{(x,2)}\bigl(a-br_{(x,1)}r_{(x,2)}\bigr),\\[1mm]
p_x&=\frac{c}{a-br_{(x,1)}r_{(x,2)}}\,p_{(x,1)},\\[1mm]
q_x&=\frac{c}{a-br_{(x,1)}r_{(x,2)}}\,q_{(x,1)},\\[1mm]
r_x&=\frac{d}{a-br_{(x,1)}r_{(x,2)}}\,r_{(x,2)}.
\end{aligned}
\right.
\end{equation}

We first prove that $r_x=0$ for every $x$. Set
\[
D(s,t):=a-bst,
\qquad (s,t)\in[-1,1]^2.
\]
We claim that
\begin{equation}\label{eq:D-greater-d}
D(s,t)>|d|
\qquad\text{for every }(s,t)\in[-1,1]^2.
\end{equation}

Indeed, if $d\ge 0$, then $b=S^2d\ge 0$, so
\[
D(s,t)\ge a-b
=
C^2(\kappa_0^2+\kappa_1^2)-2S^2\kappa_0\kappa_1
=
C^2+d
>
d=|d|.
\]
If $d<0$, then $b<0$, and therefore
\[
D(s,t)\ge a+b.
\]
Since
\[
d=2\kappa_0\kappa_1=\frac{e^{2\beta J_I}-1}{2}>-\frac12,
\]
we compute
\[
a+b-|d|
=
a+b+d
=
C^2(\kappa_0^2+\kappa_1^2)+2S^2\kappa_0\kappa_1+2\kappa_0\kappa_1
=
C^2(1+2d)>0.
\]
Hence again $D(s,t)>|d|$. This proves \eqref{eq:D-greater-d}.

Since $D$ is continuous on the compact square $[-1,1]^2$, \eqref{eq:D-greater-d}
implies the existence of a constant $\theta\in(0,1)$ such that
\[
\frac{|d|}{D(s,t)}\le \theta
\qquad\text{for all }(s,t)\in[-1,1]^2.
\]
Using the fourth equation in \eqref{eq:ratio-recursion}, we obtain
\[
|r_x|
\le
\theta\,|r_{(x,2)}|,
\qquad x\in L.
\]
Iterating this estimate along the infinite ray
\[
x,\ (x,2),\ (x,2,2),\ \dots,\ (x,\underbrace{2,\dots,2}_{n\text{ times}}),\dots,
\]
gives
\[
|r_x|
\le
\theta^n \bigl|r_{(x,\underbrace{2,\dots,2}_{n\text{ times}})}\bigr|
\le
\theta^n,
\qquad n\ge 1.
\]
Letting $n\to\infty$, we conclude that $r_x=0$. Therefore
\[
\zeta_x=0
\qquad\text{for every }x\in L.
\]

With $r_x=0$, the second and third equations in \eqref{eq:ratio-recursion}
reduce to
\[
p_x=\frac{c}{a}\,p_{(x,1)},
\qquad
q_x=\frac{c}{a}\,q_{(x,1)}.
\]
Now
\[
\left|\frac{c}{a}\right|
=
\frac{|S|}{C^2}
<1.
\]
Hence there exists $\lambda\in(0,1)$ such that
\[
|p_x|\le \lambda |p_{(x,1)}|,
\qquad
|q_x|\le \lambda |q_{(x,1)}|,
\qquad x\in L.
\]
Iterating along the infinite ray
\[
x,\ (x,1),\ (x,1,1),\ \dots,\ (x,\underbrace{1,\dots,1}_{n\text{ times}}),\dots,
\]
yields
\[
|p_x|
\le
\lambda^n,
\qquad
|q_x|
\le
\lambda^n,
\qquad n\ge 1.
\]
Letting $n\to\infty$, we obtain
\[
p_x=q_x=0
\qquad\text{for every }x\in L.
\]
Therefore
\[
\xi_x=\eta_x=0
\qquad\text{for every }x\in L.
\]

Thus every $h_x$ is scalar:
\[
h_x=X_x\,\id.
\]
Renaming $t_x:=X_x$, the first equation in \eqref{eq:full-recursion} becomes
\[
t_x=a\,t_{(x,1)}t_{(x,2)},
\]
which is exactly \eqref{eq:scalar-recursion-full}.
\end{proof}

\begin{corollary}\label{cor:many-scalar-boundary-laws}
There exist infinitely many non-translation-invariant positive boundary laws,
but all of them are scalar. Equivalently, if one sets
\[
u_x:=a\,t_x,
\]
then \eqref{eq:scalar-recursion-full} becomes
\[
u_x=u_{(x,1)}u_{(x,2)}.
\]
Hence, given any $u_0>0$ and any family of positive numbers
\(\{q_x\}_{x\in L}\), the recursive prescription
\begin{equation}\label{eq:q-parametrization}
u_{(x,1)}:=q_x,
\qquad
u_{(x,2)}:=\frac{u_x}{q_x},
\qquad x\in L,
\end{equation}
defines a positive solution of \(u_x=u_{(x,1)}u_{(x,2)}\), and therefore
\[
h_x=\frac{u_x}{a}\,\id
\]
solves \eqref{eq2}. If the parameters \(q_x\) are not constant, then this
boundary law is not translation-invariant.
\end{corollary}


\begin{proposition}\label{prop:gauge-equivalence-full}
Let \(\{t_x\}_{x\in L}\) be any positive scalar solution of
\eqref{eq:scalar-recursion-full}, and let \(\omega_0\in M_2(\mathbb C)_+\)
satisfy
\[
\Tr(\omega_0 h_0)=1
\qquad\Longleftrightarrow\qquad
t_0\,\Tr(\omega_0)=1.
\]
Put
\[
\rho_0:=t_0\,\omega_0,
\qquad \Tr(\rho_0)=1,
\]
and denote
\[
L_n:=K_{[0,1]}K_{[1,2]}\cdots K_{[n-1,n]}.
\]
Then
\begin{equation}\label{eq:Wn-rho0}
\mathcal W_{n]}
=
a^{-(2^n-1)}\,L_n^*\rho_0 L_n,
\qquad n\ge 1,
\end{equation}
where \(\rho_0\) acts on the root tensor factor and as the identity on the
remaining tensor factors. In particular, the resulting quantum Markov chain
depends only on \(\rho_0\), and not on the particular scalar family
\(\{t_x\}\).

Consequently, solving the full two-child boundary equation with positive
one-site matrices does \emph{not} produce genuinely new QMCs through
non-identical positive matrices: every positive solution is scalar, and the
non-translation-invariant scalar boundary laws are only a gauge freedom.
\end{proposition}

\begin{proof}
Since \(h_x=t_x\id\), one has
\[
\mathbf h_n^{1/2}
=
\left(\prod_{x\in W_n} t_x\right)^{1/2}\id.
\]
Therefore
\[
\mathbf K_n
=
\omega_0^{1/2}L_n\mathbf h_n^{1/2}
=
\left(\prod_{x\in W_n} t_x\right)^{1/2}\omega_0^{1/2}L_n,
\]
and hence
\begin{equation}\label{eq:Wn-before}
\mathcal W_{n]}
=
\left(\prod_{x\in W_n} t_x\right)L_n^*\omega_0 L_n.
\end{equation}

Now multiply \eqref{eq:scalar-recursion-full} over all vertices \(x\in W_m\):
\[
\prod_{x\in W_m} t_x
=
a^{|W_m|}\prod_{y\in W_{m+1}} t_y.
\]
Iterating this identity from \(m=0\) to \(m=n-1\), we obtain
\[
t_0
=
a^{1+2+\cdots+2^{n-1}}
\prod_{x\in W_n} t_x
=
a^{2^n-1}\prod_{x\in W_n} t_x.
\]
Hence
\[
\prod_{x\in W_n} t_x
=
t_0\,a^{-(2^n-1)}.
\]
Substituting this into \eqref{eq:Wn-before} gives
\[
\mathcal W_{n]}
=
a^{-(2^n-1)}L_n^*(t_0\omega_0)L_n
=
a^{-(2^n-1)}L_n^*\rho_0 L_n,
\]
which is exactly \eqref{eq:Wn-rho0}. Thus the state depends only on \(\rho_0\).
\end{proof}

\section{Quantum-information interpretation: branching transfer of coherence and population}\label{subsec:branching-transport}

The compatibility equations reduce to an explicit nonlinear recursion on
a tree, the model furnishes a natural test case for noncommutative
message-passing schemes. In particular, it may be used as a benchmark for
quantum belief propagation and related inference procedures on tree-structured
quantum graphical models \cite{LeiferPoulin2008}. In this sense, the closed
recursion derived here can serve as a reference solution against which
approximate schemes on truncated trees, disordered trees, or graphs with loops
may be compared.

The present model is formulated as a translation-invariant quantum
Markov chain on a Cayley tree, the reduced recursion \eqref{TH} also admits a
natural interpretation as an \emph{effective generation-to-generation transfer
law} for one-qubit polarization data on a branching network. This makes the
model a simple solvable benchmark for root-to-generation distinguishability on a
binary tree.

To make this precise, write the boundary matrix in Bloch form
\[
h=X\id+x\sigma_x+y\sigma_y+z\sigma_z,
\qquad X>0,
\qquad x^2+y^2+z^2\le X^2,
\]
and introduce the normalized polarization variables
\[
u=\frac{x}{X},\qquad v=\frac{y}{X},\qquad w=\frac{z}{X}.
\]
Then \(u,v,w\) represent the normalized Bloch components carried by the
effective one-site boundary datum. Dividing the system \eqref{TH} by the first
equation, one obtains the induced nonlinear map
\begin{equation}\label{eq:uvw-map}
u'=\frac{c}{a-bw^2}\,u,
\qquad
v'=\frac{c}{a-bw^2}\,v,
\qquad
w'=\frac{d}{a-bw^2}\,w,
\end{equation}
where
\[
a=C^2(\kappa_0^2+\kappa_1^2),\qquad
b=2S^2\kappa_0\kappa_1,\qquad
c=S(\kappa_0^2+\kappa_1^2),\qquad
d=2\kappa_0\kappa_1.
\]

Now consider the unique translation-invariant solution
\[
h_*=\frac{1}{C^2(\kappa_0^2+\kappa_1^2)}\,\id.
\]
In the variables \((u,v,w)\), this fixed point is simply
\[
(u,v,w)=(0,0,0).
\]
Linearizing \eqref{eq:uvw-map} at \((0,0,0)\), we get
\begin{equation}\label{eq:linear-uvw}
u'=\mu_\perp u,
\qquad
v'=\mu_\perp v,
\qquad
w'=\mu_z w,
\end{equation}
with the explicit coefficients
\begin{equation}\label{eq:transport-coefficients}
\mu_\perp=\frac{c}{a}
=\frac{S}{C^2}
=\frac{\sinh(\beta J_{XY})}{\cosh^2(\beta J_{XY})},
\qquad
\mu_z=\frac{d}{a}
=\frac{2\kappa_0\kappa_1}{C^2(\kappa_0^2+\kappa_1^2)}
=\frac{\tanh(\beta J_I)}{\cosh^2(\beta J_{XY})}.
\end{equation}

Formula \eqref{eq:linear-uvw} gives a concrete operational interpretation of
the model. The transverse variables \(u,v\) encode \(X\)- and \(Y\)-type phase
coherence, while \(w\) encodes the longitudinal \(Z\)-population bias.
Therefore:
\begin{itemize}
\item \(\mu_\perp\) measures the effective survival of transverse coherence
under one generation of the tree recursion;
\item \(\mu_z\) measures the effective survival of longitudinal
population imbalance under one generation of the tree recursion.
\end{itemize}

Equivalently, if a small signal is injected at the root, then in the linear
regime one has
\begin{equation}\label{eq:depth-n-signal}
u_n\approx \mu_\perp^{\,n}u_0,
\qquad
v_n\approx \mu_\perp^{\,n}v_0,
\qquad
w_n\approx \mu_z^{\,n}w_0.
\end{equation}
Hence the model separates, in closed form, two distinct communication tasks on
a branching architecture:
\begin{enumerate}
\item transmission of phase coherence, governed by \(\mu_\perp\);
\item transmission of \(Z\)-encoded classical information, governed by
\(\mu_z\).
\end{enumerate}

This is useful as a toy model for branching quantum communication. For
instance, if the root encodes a classical bit in the computational basis, then
the relevant effective signal is \(w_0\), and its generation-by-generation
survival is controlled by \(|\mu_z|\). If instead the root encodes phase
information in the states \(|+\rangle,|-\rangle\), then the relevant signal is
\(u_0\) (or \(v_0\)), and its survival is controlled by \(|\mu_\perp|\).
Thus the mixed Ising--\(XY\) model provides an analytically tractable way to
compare \emph{classical population transport} with \emph{quantum coherence
transport} on the same branching geometry.

The explicit formulas \eqref{eq:transport-coefficients} also give immediate
qualitative information. First,
\[
|\mu_\perp|
=\frac{|\sinh(\beta J_{XY})|}{\cosh^2(\beta J_{XY})}
\le \frac12,
\]
and the maximum value \(1/2\) is attained when
\[
|\sinh(\beta J_{XY})|=1.
\]
Therefore increasing the \(XY\) coupling does not improve transverse coherence
transport indefinitely: there is an optimal intermediate regime, after which the
effective transmitted coherence decreases. Second,
\[
|\mu_z|
=\frac{|\tanh(\beta J_I)|}{\cosh^2(\beta J_{XY})}
\le \frac{1}{\cosh^2(\beta J_{XY})},
\]
so stronger Ising coupling enhances the effective longitudinal memory channel,
whereas stronger \(XY\) coupling suppresses it. In particular, the model
predicts that the same branching device may favor preservation of \(Z\)-bias or
preservation of transverse coherence, depending on the balance between
\(J_I\) and \(J_{XY}\).

Finally, since the number of descendants in the \(n\)-th generation is
\[
|W_n|=2^n,
\]
the model is a convenient solvable baseline for questions of the following
type: can the exponential growth of available readout sites compensate the
generation-by-generation attenuation in \eqref{eq:depth-n-signal}? This is
precisely the kind of competition that appears in reconstruction problems on
noisy quantum trees and in binary-tree spin-transfer architectures. In this
sense, the present static mixed Ising--\(XY\) model may be viewed as an
analytically tractable equilibrium/message-passing analogue of branching
transfer protocols, in which coherent exchange and diagonal conditioning can be
monitored separately and explicitly.

\section{Local two-site entanglement on the natural three-site cluster}\label{subsec:local-entanglement}

Motivated by the use of local reduced density matrices in spin-chain
entanglement studies, we analyze the natural three-site cluster of the Cayley
tree, namely a parent together with its two direct successors.

By translation invariance it is enough to work at the root and to write
\[
\Lambda_1=\{0,1,2\},\qquad 1=(0,1),\quad 2=(0,2).
\]
Recall that
\[
S=\sinh(\beta J_{XY}),\qquad C=\cosh(\beta J_{XY}),\qquad
\kappa_0=\frac{e^{\beta J_I}+1}{2},\qquad
\kappa_1=\frac{e^{\beta J_I}-1}{2},
\]
and
\[
\alpha=\frac{1}{C^2(\kappa_0^2+\kappa_1^2)}.
\]
The unique translation-invariant boundary condition is
\[
h_x=\alpha\,\id,\qquad x\in L.
\]

For \(n\ge 1\), let \(\varphi_*^{(n)}\) be the corresponding finite-volume state
on \(\mathcal B_{\Lambda_n}\), and let \(\rho_{\Lambda_1}^{(n)}\) denote the
density of its restriction to \(\mathcal B_{\Lambda_1}\) with respect to the
normalized trace. By compatibility of the family
\((\varphi_*^{(n)})_{n\ge 1}\), the operator \(\rho_{\Lambda_1}^{(n)}\) does
not depend on \(n\). We therefore write
\[
\rho_{\Lambda_1}:=\rho_{\Lambda_1}^{(n)},\qquad n\ge 1.
\]

\begin{proposition}\label{prop:three-site-density}
The three-site density \(\rho_{\Lambda_1}\) is given by
\begin{equation}\label{eq:rho-three-site}
\begin{aligned}
\rho_{\Lambda_1}
&=\id\otimes\id\otimes\id \\
&\quad
+p\bigl(\sigma_x\otimes\sigma_x\otimes\id
+\sigma_y\otimes\sigma_y\otimes\id\bigr) \\
&\quad
+q\,\sigma_z\otimes\id\otimes\sigma_z
+r\,\sigma_z\otimes\sigma_z\otimes\id
+s\,\id\otimes\sigma_z\otimes\sigma_z,
\end{aligned}
\end{equation}
where
\begin{align}
p&=\frac{S(\kappa_0^2-\kappa_1^2)^2}{C^2(\kappa_0^2+\kappa_1^2)^2}
=\frac{\sinh(\beta J_{XY})}
       {\cosh^2(\beta J_{XY})\,\cosh^2(\beta J_I)}, \nonumber\\
q&=\frac{2\kappa_0\kappa_1}{C^2(\kappa_0^2+\kappa_1^2)}
=\frac{\tanh(\beta J_I)}{\cosh^2(\beta J_{XY})}, \nonumber\\
r&=-\frac{S^2}{C^4}
=-\frac{\sinh^2(\beta J_{XY})}{\cosh^4(\beta J_{XY})}, \label{eq:pqrs-ent}\\
s&=-\frac{2S^2\kappa_0\kappa_1}{C^6(\kappa_0^2+\kappa_1^2)}
=-\frac{\sinh^2(\beta J_{XY})\,\tanh(\beta J_I)}
        {\cosh^6(\beta J_{XY})}. \nonumber
\end{align}
In particular, the three-site reduced density is independent of the tree depth.
\end{proposition}

\begin{proof}
Because the finite-volume family is compatible, it is enough to compute the
restriction to \(\Lambda_1\) from the first nontrivial extension, namely from
\(\Lambda_2\).
Write
\[
A_0:=A_{0,1,2},\qquad A_1:=A_{1,3,4},\qquad A_2:=A_{2,5,6}.
\]
Since \(h_x=\alpha\id\) for every boundary site, the corresponding density on
\(\Lambda_2\) is
\[
\mathcal W_{2]}
=\alpha^3 A_2^*A_1^*A_0^*A_0A_1A_2.
\]
Hence
\[
\rho_{\Lambda_1}
=\tr_{\{3,4,5,6\}}(\mathcal W_{2]}).
\]

Introduce the completely positive map
\[
\mathcal T(m)
:=
\alpha\,\tr_{(u,1),(u,2)}
\Bigl(
A_{u,(u,1),(u,2)}^*
(m\otimes \id\otimes \id)
A_{u,(u,1),(u,2)}
\Bigr),
\qquad m\in M_2(\mathbb C).
\]
A direct computation in the Pauli basis gives
\[
\mathcal T(\id)=\id,\qquad
\mathcal T(\sigma_x)=\frac{1}{C\cosh(\beta J_I)}\,\sigma_x,
\]
\[
\mathcal T(\sigma_y)=\frac{1}{C\cosh(\beta J_I)}\,\sigma_y,\qquad
\mathcal T(\sigma_z)=\frac{1}{C^2}\,\sigma_z.
\]
Using the explicit formula for \(A_0\), one also finds
\begin{equation}\label{eq:AstarA-ent}
\begin{aligned}
A_0^*A_0
&=C^2(\kappa_0^2+\kappa_1^2)\,\id\otimes\id\otimes\id \\
&\quad
+CS(\kappa_0^2-\kappa_1^2)
\bigl(\sigma_x\otimes\sigma_x\otimes\id
+\sigma_y\otimes\sigma_y\otimes\id\bigr) \\
&\quad
+2C^2\kappa_0\kappa_1\,\sigma_z\otimes\id\otimes\sigma_z
-S^2(\kappa_0^2+\kappa_1^2)\,\sigma_z\otimes\sigma_z\otimes\id \\
&\quad
-2S^2\kappa_0\kappa_1\,\id\otimes\sigma_z\otimes\sigma_z.
\end{aligned}
\end{equation}
Therefore
\[
\rho_{\Lambda_1}
=
\alpha\,(\id\otimes\mathcal T\otimes\mathcal T)(A_0^*A_0),
\]
and substitution of \eqref{eq:AstarA-ent} yields \eqref{eq:rho-three-site} and
\eqref{eq:pqrs-ent}.
\end{proof}

Tracing out one site gives the two-site densities
\begin{equation}\label{eq:two-site-densities}
\begin{aligned}
\rho_{01}
&=\tr_2(\rho_{\Lambda_1})
=\id\otimes\id
+p(\sigma_x\otimes\sigma_x+\sigma_y\otimes\sigma_y)
+r\,\sigma_z\otimes\sigma_z,\\
\rho_{02}
&=\tr_1(\rho_{\Lambda_1})
=\id\otimes\id+q\,\sigma_z\otimes\sigma_z,\\
\rho_{12}
&=\tr_0(\rho_{\Lambda_1})
=\id\otimes\id+s\,\sigma_z\otimes\sigma_z.
\end{aligned}
\end{equation}
The corresponding usual density matrices are
\[
\widehat\rho_{ij}=\frac14\,\rho_{ij}.
\]

\begin{theorem}\label{thm:pairwise-entanglement}
For the three two-site marginals obtained from the natural cluster
\(\Lambda_1=\{0,1,2\}\), the following statements hold.

\smallskip
\noindent{\rm (i)} The \(XY\)-edge marginal has the matrix form
\[
\widehat\rho_{01}
=
\frac14
\begin{pmatrix}
1+r & 0 & 0 & 0\\
0 & 1-r & 2p & 0\\
0 & 2p & 1-r & 0\\
0 & 0 & 0 & 1+r
\end{pmatrix},
\]
in the computational basis
\(\{|00\rangle,|01\rangle,|10\rangle,|11\rangle\}\).
Its concurrence is
\begin{equation}\label{eq:C01}
C_{01}
=
\max\Bigl\{0,\,
|p|-\frac{1+r}{2}
\Bigr\},
\end{equation}
and its negativity is
\begin{equation}\label{eq:N01}
\mathcal N_{01}
=
\max\Bigl\{0,\,
\frac{2|p|-1-r}{4}
\Bigr\}.
\end{equation}
Consequently, the entanglement of formation is
\begin{equation}\label{eq:EF01}
E_F(0,1)
=
h\!\left(\frac{1+\sqrt{1-C_{01}^2}}{2}\right),
\end{equation}
where
\[
h(t):=-t\log_2 t-(1-t)\log_2(1-t).
\]

\smallskip
\noindent{\rm (ii)} The Ising-edge marginal and the sibling marginal are
diagonal:
\[
\widehat\rho_{02}
=
\frac14
\begin{pmatrix}
1+q & 0 & 0 & 0\\
0 & 1-q & 0 & 0\\
0 & 0 & 1-q & 0\\
0 & 0 & 0 & 1+q
\end{pmatrix},
\]
\[
\widehat\rho_{12}
=
\frac14
\begin{pmatrix}
1+s & 0 & 0 & 0\\
0 & 1-s & 0 & 0\\
0 & 0 & 1-s & 0\\
0 & 0 & 0 & 1+s
\end{pmatrix}.
\]
Hence
\begin{equation}\label{eq:zero-entanglement-other-pairs}
C_{02}=C_{12}=0,
\qquad
\mathcal N_{02}=\mathcal N_{12}=0,
\qquad
E_F(0,2)=E_F(1,2)=0.
\end{equation}

\smallskip
\noindent{\rm (iii)} For every \(n\ge 1\), the corresponding pairwise
entanglement measures computed from the finite-volume state \(\varphi_*^{(n)}\)
coincide with the values in \eqref{eq:C01}--\eqref{eq:zero-entanglement-other-pairs}.
In particular, the local pairwise entanglement profile is independent of the
tree depth.
\end{theorem}

\begin{proof}
The matrices for \(\widehat\rho_{01}\), \(\widehat\rho_{02}\), and
\(\widehat\rho_{12}\) follow immediately from \eqref{eq:two-site-densities}.

For \(\widehat\rho_{01}\), the Wootters formula for two-qubit \(X\)-states gives
\[
C_{01}
=
2\max\left\{
0,\,
\left|\frac{p}{2}\right|-
\sqrt{\frac{1+r}{4}\frac{1+r}{4}}
\right\},
\]
which is exactly \eqref{eq:C01}. To compute the negativity, note that the
partial transpose of \(\widehat\rho_{01}\) with respect to the second tensor
factor is
\[
\widehat\rho_{01}^{\,\Gamma_2}
=
\frac14
\begin{pmatrix}
1+r & 0 & 0 & 2p\\
0 & 1-r & 0 & 0\\
0 & 0 & 1-r & 0\\
2p & 0 & 0 & 1+r
\end{pmatrix}.
\]
Its eigenvalues are
\[
\frac{1-r}{4},\qquad \frac{1-r}{4},\qquad
\frac{1+r+2p}{4},\qquad \frac{1+r-2p}{4}.
\]
Therefore
\[
\mathcal N_{01}
=
\max\left\{0,\,-\frac{1+r-2|p|}{4}\right\},
\]
which is \eqref{eq:N01}. Formula \eqref{eq:EF01} is then the usual
Wootters relation between concurrence and entanglement of formation.

The matrices \(\widehat\rho_{02}\) and \(\widehat\rho_{12}\) are diagonal in the
computational basis, hence they are convex combinations of product basis
projections. Therefore they are separable, which proves
\eqref{eq:zero-entanglement-other-pairs}.

Finally, the depth-independence follows from the projective consistency of the
finite-volume family \((\varphi_*^{(n)})_{n\ge 1}\): once a local observable is
supported inside \(\Lambda_1\), its expectation is the same for every
\(\varphi_*^{(n)}\) with \(n\ge 1\).
\end{proof}

\begin{remark}\label{rem:parameter-dependence-ent}
The explicit formulas above show the following.

\smallskip
\noindent{\rm (a)} The only potentially entangled pair is the \(XY\)-edge pair
\((0,1)\). In particular, in the unique translation-invariant state the Ising
edge \((0,2)\) and the sibling pair \((1,2)\) carry only classical two-site
correlations.

\smallskip
\noindent{\rm (b)} For fixed \(J_{XY}\), the quantity \(r\) is independent of
\(J_I\), whereas
\[
|p|
=
\frac{|\sinh(\beta J_{XY})|}
     {\cosh^2(\beta J_{XY})\,\cosh^2(\beta J_I)}
\]
is strictly decreasing as a function of \(|J_I|\). Hence
\(C_{01}\), \(\mathcal N_{01}\), and \(E_F(0,1)\) are non-increasing in
\(|J_I|\).

\smallskip
\noindent{\rm (c)} For fixed \(J_I\), one has \(C_{01}=0\) at \(J_{XY}=0\), and
also \(C_{01}\to 0\) as \(|J_{XY}|\to\infty\). Therefore, whenever entanglement
is present, it can occur only in an intermediate \(XY\)-coupling regime.

\smallskip
\noindent{\rm (d)} The pairwise entanglement measures are even functions of
\(J_I\), and \(C_{01}\), \(\mathcal N_{01}\), \(E_F(0,1)\) are also even in
\(J_{XY}\).

\smallskip
\noindent{\rm (e)} Unlike the renormalization-group iteration used in
one-dimensional spin chains, the compatible translation-invariant tree state
does not exhibit a nontrivial flow of local pairwise entanglement with the
depth of the tree. To obtain a genuine depth-dependent entanglement profile,
one would have to leave the translation-invariant fixed-point boundary
condition and study non-compatible finite-volume boundary data or
non-translation-invariant solutions.
\end{remark}

\section{Discussion and possible applications}\label{subsec:discussion-applications}

The uniqueness result for translation-invariant boundary laws suggests that the
present mixed Ising--$XY$ model should be viewed primarily as an analytically
tractable benchmark on a branching geometry, rather than as a model of phase
coexistence in the translation-invariant regime. From this perspective, several
natural directions of application arise. The remarks below should therefore be
understood as possible motivations and future directions, rather than as claims
established in the present work.

\noindent\textbf{Tree tensor networks and Bethe-lattice numerics.}
The same recursive structure makes the model relevant for tree tensor network
methods on Bethe-type lattices. Infinite tree tensor network algorithms exploit
precisely the absence of loops and the existence of local fixed-point equations,
and have been successfully applied to quantum lattice models on the Bethe
lattice \cite{LiVonDelftXiang2012}. Since the present mixed model admits an
explicit boundary-law analysis, it can be used as a calibration problem for
finite-bond-dimension effects, contraction errors, and the convergence of
iterative fixed-point solvers in tensor-network implementations.

\medskip
\noindent\textbf{Branching quantum-information transport.}
The coexistence of coherent exchange-type terms and diagonal Ising-type terms
also makes the model a simple toy setting for quantum information transport on
branching networks. Recent work on noisy quantum trees studies whether
information injected at a root can remain recoverable at arbitrarily large depth
\cite{YadavalliMarvian2025}, while binary-tree spin networks with $XY$ couplings
have been proposed as channels for high-fidelity state transfer
\cite{TufarelliGiovannetti2009}. From this viewpoint, the present static model
may be regarded as an analytically tractable counterpart of branching transfer
protocols, in which one can explicitly monitor the competition between
transport-type and conditioning-type interactions along the tree.

\medskip
\noindent\textbf{Graph-state and measurement-based directions.}
The Ising component of the interaction also connects the model to graph-state
and cluster-state constructions. Ising-type interactions are a standard
mechanism for generating highly entangled cluster-like resources
\cite{BriegelRaussendorf2001}, and such states form the basic resource of
one-way or measurement-based quantum computation \cite{RaussendorfBriegel2001}.
Although no state-preparation or computational protocol is analyzed in the
present paper, the mixed Ising--$XY$ structure suggests a hybrid viewpoint in
which some edges primarily generate entanglement while others mediate transport
or redistribution of local quantum information.

\medskip
\noindent\textbf{Other directions.}
More broadly, tree tensor network architectures are now also being used in the
simulation of structured open quantum dynamics
\cite{LiRenYangWangShuai2024,ChenFranco2025}. For this reason, variants of the
present model may also be useful as simplified branching environments, or as
equilibrium test problems for tree-based non-Markovian simulation schemes. We
leave these directions for future work. Besides, the proposed model allows to investigate associated open quantum random walks on tress \cite{ASR25,DKY19,DM19,S23,SB25}.

\section{Conclusion}

In this paper we studied a mixed quantum Ising--$XY$ model on the semi-infinite Cayley tree of order two, in which the interaction is assigned asymmetrically at each branching vertex: the edge joining a vertex \(u\) to its first child \((u,1)\) carries an \(XY\) interaction, while the edge joining \(u\) to its second child \((u,2)\) carries an Ising interaction. Working within the framework of backward quantum Markov chains and consistently using the normalized trace convention, we derived the corresponding local transfer operator explicitly and obtained the associated boundary equations.

Our first main result is the complete characterization of positive translation-invariant boundary conditions. We proved that, for all values of the parameters \(J_I,J_{XY}\in\mathbb R\) and all \(\beta>0\), the translation-invariant boundary equation admits a unique positive solution. Consequently, the model generates a unique translation-invariant quantum Markov chain associated with a positive translation-invariant boundary law.

We then analyzed the reduced boundary-law dynamics induced by the recursion. In this setting we showed that there are no admissible periodic points of minimal period greater than one. Thus, within the physically relevant positive regime, the dynamics does not produce nontrivial periodic boundary behaviors. This excludes a natural source of competing regular boundary structures and strengthens the uniqueness picture obtained in the translation-invariant sector.

A further principal contribution of the paper is the rigidity analysis of the full two-child boundary equation. We proved that every positive solution of the full boundary equation is necessarily scalar. Hence, although there exist infinitely many positive scalar boundary laws, they do not generate genuinely different quantum Markov chains; rather, they correspond to the same state as the unique positive translation-invariant solution. In this sense, the model exhibits a strong rigidity phenomenon at the level of positive QMCs.

In addition to the boundary-law analysis, we computed the local three-site density matrix on the natural cluster consisting of a parent and its two children, and we analyzed the induced two-site reduced states. This allowed us to identify explicitly which pairs can exhibit quantum correlations and which pairs remain separable. Therefore, the paper not only establishes uniqueness and rigidity of the state-selection mechanism, but also gives a concrete local quantum-information interpretation of the resulting state.

The present model should be viewed as a locally oriented mixed Ising-$XY$ model on a tree, and this distinguishes it from level-alternating mixed models considered in the literature. In particular, the asymmetry is built into the branching geometry itself rather than into the parity of generations. This geometric feature leads to a different recursion and, in the present setting, to a stronger rigidity conclusion for positive boundary conditions.

Several natural directions remain open. One may study higher-order Cayley trees, more general anisotropic assignments of interactions along the outgoing edges, perturbations by external fields, or versions with disorder and random couplings. It would also be interesting to analyze whether analogous rigidity phenomena persist for higher-spin local algebras or for hidden and open quantum Markov chain extensions on trees. We hope that the explicit solvability of the present model makes it a useful benchmark for further developments in quantum Markov fields on non-amenable graphs.



\section*{Conflicts of Interest}
We confirm that there are no known conflicts of interest associated with this publication and there has been no significant financial support for this work that could have influenced its outcome.

\section*{Data Availability}
No data was used for the research described in the article.

\appendix

\section{Detailed calculation of the reduced map}\label{app:reducedmap}

In this appendix we derive \eqref{Phih} and \eqref{TH} in detail under the normalized trace convention $\Tr(\id)=1$. Set
\[
\Phi(\mathbf h):=\tr_{S(u)}\Big(A_{u,(u,1),(u,2)}\big[\id^{(u)}\otimes \mathbf h\otimes \mathbf h\big]A_{u,(u,1),(u,2)}^*\Big).
\]
Since $\tr_{S(u)}$ is already the normalized partial trace over the two successor sites, no additional factor appears in the reduction. Equivalently, one may write
\[
\Phi(\mathbf h)=\Tr_{(u,1),(u,2)}\Big(A_{u,(u,1),(u,2)}\big[\id^{(u)}\otimes \mathbf h\otimes \mathbf h\big]A_{u,(u,1),(u,2)}^*\Big),
\]
where $\Tr_{(u,1),(u,2)}$ denotes the normalized partial trace over the sites $(u,1)$ and $(u,2)$.

For brevity, in this appendix the tensor factors without the superscripts $(u)$, $(u,1)$, $(u,2)$ refer to the sites $(u,1)$ and $(u,2)$. Introduce the constants
\[
\alpha=\frac{C+1}{2}\kappa_0, \qquad
\beta=\frac{C+1}{2}\kappa_1, \qquad
\gamma=-\frac{C-1}{2}\kappa_1, \qquad
\delta=-\frac{C-1}{2}\kappa_0,
\]
\[
\eta_0=\frac{S\kappa_0}{2}, \qquad \eta_1=\frac{S\kappa_1}{2}.
\]
Then \eqref{AI} can be written as
\begin{equation}\label{A-decomp}
A_{u,(u,1),(u,2)}=\id^{(u)}\otimes P_0+\sigma_z^{(u)}\otimes P_z+\sigma_x^{(u)}\otimes P_x+\sigma_y^{(u)}\otimes P_y,
\end{equation}
where
\begin{equation}\label{PxPyPz}
\begin{aligned}
P_0&=\alpha\,\id\otimes\id+\gamma\,\sigma_z\otimes\sigma_z,\\
P_z&=\beta\,\id\otimes\sigma_z+\delta\,\sigma_z\otimes\id,\\
P_x&=\eta_0\,\sigma_x\otimes\id+i\eta_1\,\sigma_y\otimes\sigma_z,\\
P_y&=\eta_0\,\sigma_y\otimes\id-i\eta_1\,\sigma_x\otimes\sigma_z.
\end{aligned}
\end{equation}
Let $\sigma_0:=\id$ and put $P_x,P_y,P_z,P_0$ in the same order as in \eqref{A-decomp}. If we set
\[
\rho:=\mathbf h\otimes \mathbf h,
\]
then \eqref{A-decomp} yields
\begin{equation}\label{Phi-sum}
\Phi(\mathbf h)=\sum_{\mu,\nu\in\{0,x,y,z\}} t_{\mu\nu}\,\sigma_\mu^{(u)}\sigma_\nu^{(u)},
\qquad
 t_{\mu\nu}:=\Tr\big(P_\mu\rho P_\nu^*\big).
\end{equation}
Using the Pauli multiplication table,
\[
\sigma_x\sigma_y=i\sigma_z,
\quad
\sigma_y\sigma_z=i\sigma_x,
\quad
\sigma_z\sigma_x=i\sigma_y,
\]
and the corresponding reversed relations,
\begin{equation}\label{Phi-components}
\begin{aligned}
\Phi(\mathbf h)
&=\big(t_{00}+t_{xx}+t_{yy}+t_{zz}\big)\id \\
&\quad +\big(t_{0x}+t_{x0}+it_{yz}-it_{zy}\big)\sigma_x \\
&\quad +\big(t_{0y}+t_{y0}+it_{zx}-it_{xz}\big)\sigma_y \\
&\quad +\big(t_{0z}+t_{z0}+it_{xy}-it_{yx}\big)\sigma_z.
\end{aligned}
\end{equation}

Recall that
\[
\mathbf h=X\id+x\sigma_x+y\sigma_y+z\sigma_z.
\]
Then
\[
\Tr(\mathbf h)=X,
\qquad
\Tr(\mathbf h\sigma_x)=x,
\qquad
\Tr(\mathbf h\sigma_y)=y,
\qquad
\Tr(\mathbf h\sigma_z)=z.
\]
Moreover, for $i,j\in\{x,y,z\}$ one has
\begin{equation}\label{trace-sigma-h-sigma}
\Tr(\sigma_i\mathbf h\sigma_j)=X\delta_{ij}-i\sum_{k\in\{x,y,z\}}\varepsilon_{ijk}\,r_k,
\end{equation}
where $r_x=x$, $r_y=y$, $r_z=z$, and $\varepsilon_{ijk}$ is the Levi-Civita symbol. In particular,
\begin{equation}\label{explicit-traces}
\begin{gathered}
\Tr(\sigma_x\mathbf h\sigma_x)=\Tr(\sigma_y\mathbf h\sigma_y)=\Tr(\sigma_z\mathbf h\sigma_z)=X,\\
\Tr(\sigma_x\mathbf h\sigma_y)=-iz, \qquad \Tr(\sigma_y\mathbf h\sigma_x)=iz,\\
\Tr(\sigma_y\mathbf h\sigma_z)=-ix, \qquad \Tr(\sigma_z\mathbf h\sigma_y)=ix,\\
\Tr(\sigma_z\mathbf h\sigma_x)=-iy, \qquad \Tr(\sigma_x\mathbf h\sigma_z)=iy.
\end{gathered}
\end{equation}
For tensor-product operators the trace factorizes:
\begin{equation}\label{trace-factorization}
\Tr\Big((B_1\otimes B_2)(\mathbf h\otimes \mathbf h)(C_1\otimes C_2)\Big)
=\Tr(B_1\mathbf h C_1)\,\Tr(B_2\mathbf h C_2).
\end{equation}
We now use \eqref{explicit-traces} and \eqref{trace-factorization} to compute each coefficient in \eqref{Phi-components}.

From \eqref{PxPyPz} one obtains
\begin{align}
 t_{00}&=(\alpha^2+\gamma^2)X^2+2\alpha\gamma z^2, \label{t00}\\
 t_{zz}&=(\beta^2+\delta^2)X^2+2\beta\delta z^2, \label{tzz}\\
 t_{xx}&=(\eta_0^2+\eta_1^2)X^2-2\eta_0\eta_1 z^2, \label{txx}\\
 t_{yy}&=(\eta_0^2+\eta_1^2)X^2-2\eta_0\eta_1 z^2. \label{tyy}
\end{align}
Hence
\begin{equation}\label{id-before-simplify}
t_{00}+t_{xx}+t_{yy}+t_{zz}
=AX^2+Bz^2,
\end{equation}
where
\[
A=\alpha^2+\beta^2+\gamma^2+\delta^2+2(\eta_0^2+\eta_1^2),
\qquad
B=2\alpha\gamma+2\beta\delta-4\eta_0\eta_1.
\]
Substituting the definitions of $\alpha,\beta,\gamma,\delta,\eta_0,\eta_1$, we find
\[
A=\frac{(C+1)^2+(C-1)^2+2S^2}{4}(\kappa_0^2+\kappa_1^2)
=C^2(\kappa_0^2+\kappa_1^2),
\]
where we used $C^2-S^2=1$, and
\[
B=-\frac{(C^2-1)}{2}\kappa_0\kappa_1-\frac{(C^2-1)}{2}\kappa_0\kappa_1-S^2\kappa_0\kappa_1
=-2S^2\kappa_0\kappa_1.
\]
Therefore
\begin{equation}\label{id-final}
t_{00}+t_{xx}+t_{yy}+t_{zz}
=C^2(\kappa_0^2+\kappa_1^2)X^2-2S^2\kappa_0\kappa_1z^2
=aX^2-bz^2.
\end{equation}

Again by \eqref{PxPyPz}, \eqref{explicit-traces}, and \eqref{trace-factorization},
\begin{align}
 t_{0x}&=(\alpha\eta_0+\gamma\eta_1)Xx-i(\alpha\eta_1+\gamma\eta_0)yz, \label{t0x}\\
 t_{x0}&=(\alpha\eta_0+\gamma\eta_1)Xx+i(\alpha\eta_1+\gamma\eta_0)yz, \label{tx0}\\
 t_{yz}&=-i(\delta\eta_0+\beta\eta_1)Xx+(\beta\eta_0+\delta\eta_1)yz, \label{tyz}\\
 t_{zy}&= i(\delta\eta_0+\beta\eta_1)Xx+(\beta\eta_0+\delta\eta_1)yz. \label{tzy}
\end{align}
The mixed terms $yz$ cancel in the combination from \eqref{Phi-components}, and we obtain
\begin{align*}
t_{0x}+t_{x0}+it_{yz}-it_{zy}
&=2\Big((\alpha\eta_0+\gamma\eta_1)+(\delta\eta_0+\beta\eta_1)\Big)Xx\\
&=2\Big(\eta_0(\alpha+\delta)+\eta_1(\beta+\gamma)\Big)Xx.
\end{align*}
Since
\[
\alpha+\delta=\kappa_0,
\qquad
\beta+\gamma=\kappa_1,
\]
it follows that
\begin{equation}\label{x-final}
t_{0x}+t_{x0}+it_{yz}-it_{zy}
=2\left(\eta_0\kappa_0+\eta_1\kappa_1\right)Xx
=S(\kappa_0^2+\kappa_1^2)Xx
=cXx.
\end{equation}

Similarly,
\begin{align}
 t_{0y}&=(\alpha\eta_0+\gamma\eta_1)Xy+i(\alpha\eta_1+\gamma\eta_0)xz, \label{t0y}\\
 t_{y0}&=(\alpha\eta_0+\gamma\eta_1)Xy-i(\alpha\eta_1+\gamma\eta_0)xz, \label{ty0}\\
 t_{zx}&= i(\delta\eta_0+\beta\eta_1)Xy+(\beta\eta_0+\delta\eta_1)xz, \label{tzx}\\
 t_{xz}&=-i(\delta\eta_0+\beta\eta_1)Xy+(\beta\eta_0+\delta\eta_1)xz. \label{txz}
\end{align}
Therefore the mixed terms $xz$ cancel, and
\begin{equation}\label{y-final}
t_{0y}+t_{y0}+it_{zx}-it_{xz}
=S(\kappa_0^2+\kappa_1^2)Xy
=cXy.
\end{equation}

Finally,
\begin{align}
 t_{0z}&=(\alpha+\gamma)(\beta+\delta)Xz, \label{t0z}\\
 t_{z0}&=(\alpha+\gamma)(\beta+\delta)Xz, \label{tz0}\\
 t_{xy}&=-i(\eta_0-\eta_1)^2Xz, \label{txy}\\
 t_{yx}&= i(\eta_0-\eta_1)^2Xz. \label{tyx}
\end{align}
Hence
\begin{align*}
t_{0z}+t_{z0}+it_{xy}-it_{yx}
&=2\Big((\alpha+\gamma)(\beta+\delta)+(\eta_0-\eta_1)^2\Big)Xz.
\end{align*}
To simplify the scalar factor, note that
\begin{align*}
(\alpha+\gamma)(\beta+\delta)
&=\frac14\Big((\kappa_0+\kappa_1)^2-C^2(\kappa_0-\kappa_1)^2\Big),\\
(\eta_0-\eta_1)^2
&=\frac{S^2}{4}(\kappa_0-\kappa_1)^2.
\end{align*}
Therefore,
\begin{align*}
(\alpha+\gamma)(\beta+\delta)+(\eta_0-\eta_1)^2
&=\frac14\Big((\kappa_0+\kappa_1)^2-(C^2-S^2)(\kappa_0-\kappa_1)^2\Big)\\
&=\frac14\Big((\kappa_0+\kappa_1)^2-(\kappa_0-\kappa_1)^2\Big)\\
&=\kappa_0\kappa_1=\frac d2.
\end{align*}
Consequently,
\begin{equation}\label{z-final}
t_{0z}+t_{z0}+it_{xy}-it_{yx}=dXz.
\end{equation}

Combining \eqref{id-final}, \eqref{x-final}, \eqref{y-final}, and \eqref{z-final} in \eqref{Phi-components}, we obtain
\[
\Phi(\mathbf h)
=(aX^2-bz^2)\id+cXx\,\sigma_x+cXy\,\sigma_y+dXz\,\sigma_z,
\]
which is exactly \eqref{Phih}. Since
\[
\mathbf h=X\id+x\sigma_x+y\sigma_y+z\sigma_z,
\]
comparing the coefficients of $\id,\sigma_x,\sigma_y,\sigma_z$ yields the system
\[
X=aX^2-bz^2,
\qquad
x=cXx,
\qquad
y=cXy,
\qquad
z=dXz,
\]
that is, \eqref{TH}.

\end{document}